\documentclass[prx,twocolumn,english,superscriptaddress,floatfix,longbibliography]{revtex4-2}
\usepackage{array,amsmath,amssymb,amsthm,tabularx,tikz,tikz-cd,braket,xcolor,appendix,quantikz,algorithm,algpseudocode}
\usepackage[margin=1in]{geometry}
\usepackage{hyperref}
\definecolor{mylinkblue}{HTML}{0000EE} % classic link blue
\hypersetup{
  colorlinks=true,
  linkcolor=mylinkblue,   % refs to sections/figures
  citecolor=mylinkblue,   % \cite
  urlcolor=mylinkblue,    % \url and \href{http://...}
  anchorcolor=mylinkblue
}
\usepackage{algorithmicx}
\usepackage{algpseudocode}
\usetikzlibrary{arrows.meta,decorations.pathmorphing,positioning}
\newtheorem{lemma}{Lemma}
\newtheorem{theorem}{Theorem}

\newtheorem{definition}{Definition}

\begin{document}

\title{Compile-once block encodings for masked similarity-transformed effective Hamiltonians}

\author{Bo Peng}
\email{peng398@pnnl.gov}
\affiliation{Physical and Computational Sciences Directorate, Pacific Northwest National Laboratory, Richland, Washington 99354, United States}

\author{Yuan Liu}
\affiliation{Department of Electrical and Computer Engineering, Department of Computer Science, Department of Physics, North Carolina State University, Raleigh, North Carolina 27695, USA}

\author{Karol Kowalski}
\affiliation{Physical and Computational Sciences Directorate, Pacific Northwest National Laboratory, Richland, Washington 99354, United States}
\date{\today}

\begin{abstract}
\noindent We present \texttt{COMPOSER}, a compile-once modular parametric oracle for similarity-encoded effective reduction of electronic-structure operators (e.g., Schrieffer--Wolff-type constructions). Low-rank factorizations compress Hamiltonians and anti-Hermitian generators into rank-one bilinear and projected-quadratic ladders with {\color{black}near-linear scaling at fixed thresholds}; each ladder admits deterministic, number-conserving preparation and a block encoding using constant number of signal ancillas. A fixed \texttt{PREP-SELECT-PREP}$^\dagger$ template multiplexes these ladders, and one QSP polynomial performs the spectral transformation with degree set by operator norms. {\color{black}For a fixed orbital pool and qubit register,} the two-qubit fabric is compiled once; geometry, active-space {\color{black}(mask)} updates, and truncations are absorbed by re-dialed single-qubit rotations. We introduce a mask-aware similarity-sandwich effective-Hamiltonian construction and benchmark stability under low-rank and second-order-perturation-guided screening. \texttt{COMPOSER} is an execution architecture: algorithmic errors (block-encoding and QSP approximation) are tunable for any supplied parameters, while physical accuracy depends on how those parameters are obtained if not refined.
\end{abstract}

\maketitle

\tableofcontents

%%%%%%%%%%%%%%%%%%%%%%%%%%%%%%%%%

\section{Introduction}

Quantum computation promises to transform how we model and understand many-body quantum systems, particularly in regimes where classical approaches encounter exponential scaling bottlenecks. In quantum chemistry and materials science, core tasks such as computing ground- and excited-state energies, simulating real-time dynamics, and extracting mechanistic insight from correlated electronic structure remain computationally demanding even on the most advanced classical platforms \cite{Lloyd1996UniversalQuantumSimulators,AbramsLloyd1997FermiSimulation,Georgescu2014QuantumSimulationRMP,AspuruGuzik2005ScienceQChem,Reiher2017Elucidating,vonBurg2020Catalysis,Cao2019QuantumChemistryReview,McArdle2020RMP,Bauer2020ChemRev,McClean2020OpenFermion,vonBurg2022THC}. Beyond static ground-state properties, there is growing interest in quantum algorithms for dynamics, spectroscopy, transport, and temperature-dependent phenomena, which place additional demands on both algorithmic design and circuit-level efficiency \cite{Kassal2008ChemicalDynamics,Whitfield2011ElectronicStructureSim,Bidart2025QuantumChemReview,Tilly2022PhysRep,Berry2025QFMM,daJornada2025ComprehensiveFramework}.
Beyond phase-estimation-based approaches, NISQ-era methods such as quantum imaginary-time evolution and Krylov/Lanczos-type procedures provide alternatives for ground- and excited-state estimation \cite{Motta2020QITE,Seki2020QuantumLanczos}.

As fault-tolerant quantum computing (FTQC) architectures begin to take shape, attention has increasingly shifted from asymptotic algorithmic scaling alone toward the full execution stack required for scientifically relevant workloads. This includes state preparation, operator encoding, time evolution, and observable estimation, as well as the classical--quantum feedback loops that arise in adaptive and iterative workflows \cite{Preskill2018NISQ,Wecker2014GateCountQChem,Kandala2017HardwareEfficientVQE,Bharti2022NISQReview,Cerezo2021VQAReview,Endo2021ErrorMitigationReview,Sugisaki2025QSCI,Zhang2025FaultTolerantSurvey,daJornada2025ComprehensiveFramework}. In this setting, circuit compilation, data loading, and reuse across closely related problem instances can dominate wall-clock cost and strongly influence practical feasibility.
A central tension therefore emerges: electronic-structure operators possess rich algebraic structure, such as low-rank factorizations, excitation hierarchies, and subspace locality, while quantum circuit implementations typically flatten these structures into instance-specific gate sequences. As a result, closely related problem instances often trigger structural recompilation of multiplexing trees and routing layers, even when the underlying operator algebra changes only parametrically.
{\color{black}
A second, closely related source of recompilation overhead arises in \emph{subspace diagonalization} and \emph{subspace expansion} algorithms. 
In this family of methods one constructs an effective Hamiltonian and overlap matrix,
$H_{ij}=\langle \psi_i|\hat H|\psi_j\rangle$ and $S_{ij}=\langle \psi_i|\psi_j\rangle$, in a (generally non-orthogonal) basis of parametrized states $\{|\psi_i\rangle\}_{i=1}^K$ and solves a generalized eigenvalue problem to estimate energies.
Representative examples include quantum subspace expansion (QSE)~\cite{McClean2017QSE}, non-orthogonal VQE~\cite{Huggins2020NonOrthogonalVQE}, and NOQE-style approaches that avoid on-device optimization~\cite{Baek2023NOQE}. 
Our recent generator-coordinate-inspired methods (GCIM) and adaptive variants further emphasize the need to repeatedly \emph{add} basis states and re-evaluate matrix elements as the working subspace grows~\cite{Zheng2023QuantumGCM,Zheng2024UnleashedGCIM}. 
Across such updates, the underlying state-preparation circuits typically differ only by which generators are activated and by their scalar coefficients, making them an ideal target for a compile-once, mask-aware execution model.
}

{\color{black}
More broadly, many scientific workflows require \emph{effective} descriptions obtained by integrating out high-energy degrees of freedom or restricting to a chemically relevant model space. In the Schrieffer--Wolff (SW) paradigm~\cite{SchriefferWolff1966,Bravyi2011SchriefferWolff} one seeks a near-identity unitary $e^{\hat\sigma}$ (with anti-Hermitian $\hat\sigma$) that approximately block-diagonalizes the Hamiltonian $\hat H$ with respect to projectors $P$ and $Q=1-P$, leading to an effective operator of the form $H_{\mathrm{eff}} = P e^{-\hat\sigma}\hat H e^{\hat\sigma} P$ (and related projection-based variants)~\cite{Okubo1954Diagonalization,BlochHorowitz1958,Feshbach1958Projection}. Such effective-Hamiltonian constructions underpin downfolding and embedding strategies across chemistry and materials science~\cite{YanaiChan2007CanonicalTransformation,Knizia2012DMET,Georges1996DMFT,Wesolowski1993FDE,JacobNeugebauer2014SubsystemDFT}, including projector-based formalisms and canonical/similarity transformations. Importantly, SW can be viewed as a controlled low-order (or single-step) limit of \emph{continuous} unitary flow-equation / similarity-renormalization approaches~\cite{Wegner1994FlowEq,GlazekWilson1993RenormHam,Kehrein2006FlowEqBook}, where the generator and truncation masks are updated iteratively while the operator basis is kept structured. In these settings, the operator algebra may remain stable while only numerical data (coefficients, thresholds, or masks) changes across a family of related instances.}

At the heart of many quantum algorithms for chemistry and many-body physics lies Hamiltonian simulation: the problem of encoding and evolving a high-dimensional operator with controlled error and resource overhead. Foundational approaches such as linear-combination-of-unitaries methods, quantum phase estimation, and qubitization have established optimal asymptotic bounds and motivated the modern framework of block encoding and polynomial spectral transformations, including quantum signal processing (QSP) and quantum singular value transformation (QSVT) \cite{Somma2002SimulatingPhysicalPhenomena,AbramsLloyd1999EigenvalueEstimation,Childs2012LCU,Berry2015Taylor,LowChuang2017QSP,Gilyen2019QSVT,Low2019Qubitization,Childs2018QuantumSimulationSpeedup}. Product-formula (Trotter--Suzuki) simulation remains an important complementary approach, with modern commutator-scaling error analyses \cite{Childs2021TrotterError}. Subsequent work has extended these ideas to fault-tolerant settings, structured Hamiltonians, and spectrum amplification techniques, further reinforcing block encoding as a unifying abstraction layer between physics and hardware \cite{Low2025SpectrumAmplification,Rocca2024SCDF,Turner2025BlockEncodingStructured,Liu2024EfficientBlockEncoding}.

In parallel, a large body of work has focused on reducing constant factors by exploiting physical structure in electronic-structure Hamiltonians. It is now well established that molecular Hamiltonians admit aggressive low-rank factorizations of the two-electron tensor, including pivoted Cholesky decompositions, tensor hypercontraction, and nested singular-value decompositions \cite{Whitten1973RI,Vahtras1993RI,BeebeLinderberg1977,Aquilante2007Cholesky,RoeeggenJohansen2008Cholesky,2020DFReview,Hohenstein2012THC1,Hohenstein2012THC2,Peng2017Cholesky,Berry2019ArbitraryBasis,motta2020lowrank}. These representations compress the Hamiltonian into collections of rank-one or low-rank operators whose total count scales nearly linearly with system size for chemically relevant thresholds, enabling deterministic, number-conserving circuit constructions based on Givens rotations and fermionic swap networks \cite{Kivlichan2018FermionicSwap,Babbush2018LowDepth,Babbush2018EncodingElectronicSpectra}.

Related low-rank structure also arises in correlated wavefunction methods. Coupled-cluster and unitary coupled-cluster generators, particularly at the doubles level, can be reshaped and factorized into rank-one excitation channels using singular-value and eigen-decomposition techniques \cite{Parrish2019RankReducedCC,Hohenstein2022RankReducedCC,Peruzzo2014VQE,McClean2016TheoryVQE,OMalley2016ScalableQChem,McClean2017QSE}. These decompositions reveal a shared algebraic backbone between Hamiltonians and similarity-transformation generators, suggesting that both objects can be treated within a unified operator-encoding framework. At the same time, perturbative methods such as MP2~\cite{MollerPlesset1934,helgaker2000mp2} often capture the dominant excitation subspaces at low cost, providing natural guidance for truncation, screening, and adaptive refinement strategies.

Despite these advances, most existing quantum simulation workflows still treat each block encoding as an instance-specific artifact. Changes in molecular geometry, basis set, active space, or truncation thresholds typically require recompiling large portions of the circuit, even when the underlying operator algebra is unchanged. This recompilation overhead complicates geometry scans, active-space growth, and adaptive downfolding strategies, and obscures the separation between circuit topology and numerical data that is central to scalable quantum software design \cite{Nam2020GroundStateCompilation,Cowtan2020QubitRouting,Sivarajah2020tket,Zulehner2018MappingIBM,Grimsley2019ADAPTVQE}. The conceptual reformulation pursued here—shifting from flat circuit execution to explicit operator-level structure control—is illustrated schematically in Figure~\ref{fig:motivation}.

Meanwhile, several end-to-end quantum simulation frameworks have emerged that aim to integrate state preparation, Hamiltonian encoding, and time evolution into cohesive pipelines suitable for early fault-tolerant devices \cite{Berry2025QFMM,daJornada2025ComprehensiveFramework}. These efforts emphasize that state preparation and data access can be resource drivers comparable to Hamiltonian simulation itself, motivating continued work on deterministic loading schemes, eigenstate preparation, and efficient amplitude encoding \cite{Berry2018EigenstatePrep,GundlapalliLee2021StatePrep,Park2019QRAM}. However, these frameworks typically optimize resource counts for a given problem instance, whereas how to preserve a fixed logical circuit topology across a sequence of related instances encountered in geometry scans, active-space growth, or adaptive similarity transformations is still an unexplored area. In particular, the explicit decoupling of fixed circuit topology from instance-dependent numerical parameters is often implicit rather than an architectural principle.

This work addresses that gap by introducing \texttt{COMPOSER} (Compile-Once Modular Parametric Oracle for Similarity-Encoded Effective Reduction), a block-encoding architecture designed to explicitly separate circuit topology from numerical data. {\color{black}From an SW perspective, \texttt{COMPOSER} provides a compile-once execution stack for repeatedly constructing and updating block encodings of effective Hamiltonians induced by masked similarity transformations and projected model spaces.}
Starting from nested low-rank factorizations of molecular Hamiltonians and anti-Hermitian similarity generators, \texttt{COMPOSER} reduces both objects to linearly scaling collections of rank-one bilinear and projected-quadratic ladders. Each ladder is implemented using a deterministic, number-conserving adaptor circuit and wrapped into a single-ancilla block encoding. These adaptors are assembled within a fixed \texttt{PREP-SELECT-PREP}$^\dagger$ skeleton whose two-qubit connectivity and ancilla usage are synthesized once and then frozen. {\color{black}For any supplied set of coefficients and masks, the resulting block-encoding and polynomial-approximation errors are systematically controllable via factorization thresholds, rotation precision, and polynomial degree.}

We emphasize that \texttt{COMPOSER} is not a new Hamiltonian-simulation algorithm in the asymptotic sense. Instead, it is an architectural execution model for block-encoded chemistry operators: the two-qubit multiplexing, routing, and signal-processing skeleton is synthesized once and then reused verbatim, while updates to molecular geometry, basis/active space choices, and classical truncation masks enter only through re-dialed single-qubit rotations and classical control data. 

%\textit{What is actually new in ``compile--once''?}
Although many block-encoding constructions reuse a generic \texttt{PREP-SELECT-PREP}$^\dagger$ scaffold, the \emph{compile--once} claim in \texttt{COMPOSER} is about \emph{topology invariance across a family of related downfolded problems}: a single compiled two-qubit fabric supports geometry updates, active-space/model-space masks, and generator truncation masks inside a similarity-sandwiched effective Hamiltonian.
This is enabled by three coupled design choices: (1) a \emph{mask-aware similarity sandwich} that treats $P^{(m)}e^{-\hat\sigma^{(m)}}\hat H e^{\hat\sigma^{(m)}}P^{(m)}$ as the primary encoded object; (2) a \emph{shared rank-one ladder operator language} for both $\hat H$ and $\hat\sigma$, so the same adaptor bank applies to Hamiltonians and generators; and (3) \emph{deterministic, number-conserving ladder realizations} (Givens/pair--Givens plus fixed routing) whose two-qubit pattern is fixed once an orbital pool and pivot choices are fixed. As a result, instance dependence is isolated to streamed coefficients and single-qubit rotation angles, while the selector tree, adaptor wiring, and QSP scaffold are compiled once and reused verbatim.

The current ``compile-once'' direction directly targets the problems encountered in related approaches. For example, in qROM/data-lookup LCU implementations, the data-access structure (e.g., qROM layout and its multiplexers) is typically built for a specific coefficient table; changing truncation patterns or term lists commonly triggers rebuilding the data-loading circuitry even when the high-level template is unchanged \cite{Park2019QRAM}. In Pauli-sum/Trotter pipelines, instance updates change the Pauli list and its grouping/ordering, so simulation and measurement circuits are regenerated and re-routed \cite{Childs2021TrotterError}. In variational settings, ``compile-once'' can refer to a fixed ansatz skeleton with updated parameters, but this targets state-preparation circuits and does not by itself enforce oracle-level topology invariance under active-space and similarity-generator masks \cite{Kandala2017HardwareEfficientVQE,Cerezo2021VQAReview}.

While the qubitization- and LCU-based constructions establish optimal asymptotic query and gate complexities for fixed Hamiltonians, they typically treat each block encoding as an instance-specific object whose circuit realization is synthesized for a particular operator decomposition. In contrast, \texttt{COMPOSER} promotes topology invariance to a design principle: the multiplexed two-qubit circuit fabric implementing the block encoding is compiled once, and subsequent changes in operator coefficients, truncation masks, or active spaces are absorbed entirely into updated single-qubit parameters and classical selector data.

Within this fixed topology, all instance dependence—including molecular geometry, basis choice, active-space expansion, truncation masks, and similarity-transformation targets—enters solely through single-qubit rotation angles. 
{\color{black} Importantly, these angles need not be sourced from coupled-cluster amplitudes: for Hamiltonian block encodings they are typically determined by numerically controlled integral factorizations (e.g., Cholesky/density fitting/tensor hyper-contraction ranks and thresholds), while coupled-cluster or tensor-network amplitudes provide only one convenient \emph{initialization} for optional similarity generators or masking heuristics. For any fixed choice of encoded operators, the residual block-encoding and QSP approximation errors are independently and systematically controllable via factorization rank/thresholds, rotation precision, and polynomial degree. {\color{black}If all angles are sourced from a classical approximate object and never refined, then the achievable physical accuracy is ultimately bounded by that approximation even when the algorithmic simulation error is tunable.}}
A single quantum signal processing polynomial implements time evolution {\color{black}(and more general spectral transformations)} for the resulting masked operator sums, with polynomial degree determined by spectral norms rather than term count \cite{MotlaghWiebe2024GeneralizedQSP,BerryMotlaghWiebe2024DoublingGQSP}. {\color{black}Similarity updates enter through re-parameterizing the encoded operator (the mask-aware similarity-sandwich), not by treating conjugation itself as a spectral function of the Hamiltonian.} This compile-once, tune-many design enables adaptive workflows with predictable resource bounds while eliminating recompilation of the fixed two-qubit circuit topology across closely related problem instances, with only single-qubit parameters updated between instances.

The focus of this work is the construction and reuse of a fixed-topology block-encoding architecture---and the resulting mask-aware, compile-once execution model---rather than proposing a new simulation primitive or reporting end-to-end chemical accuracy benchmarks. We introduce deterministic rank-one adaptor circuits, develop a mask-aware similarity-sandwich construction for effective Hamiltonians, analyze depth and ancilla tradeoffs, and validate the stability of the underlying rank-one operator structure using numerical experiments and MP2-guided subspace selection. Together, these elements establish \texttt{COMPOSER} as a hardware-native block-encoding framework that supports efficient adaptive quantum simulation in both NISQ demonstrations and early fault-tolerant regimes.
\begin{figure*}[t]
\centering
\includegraphics[width=0.96\textwidth]{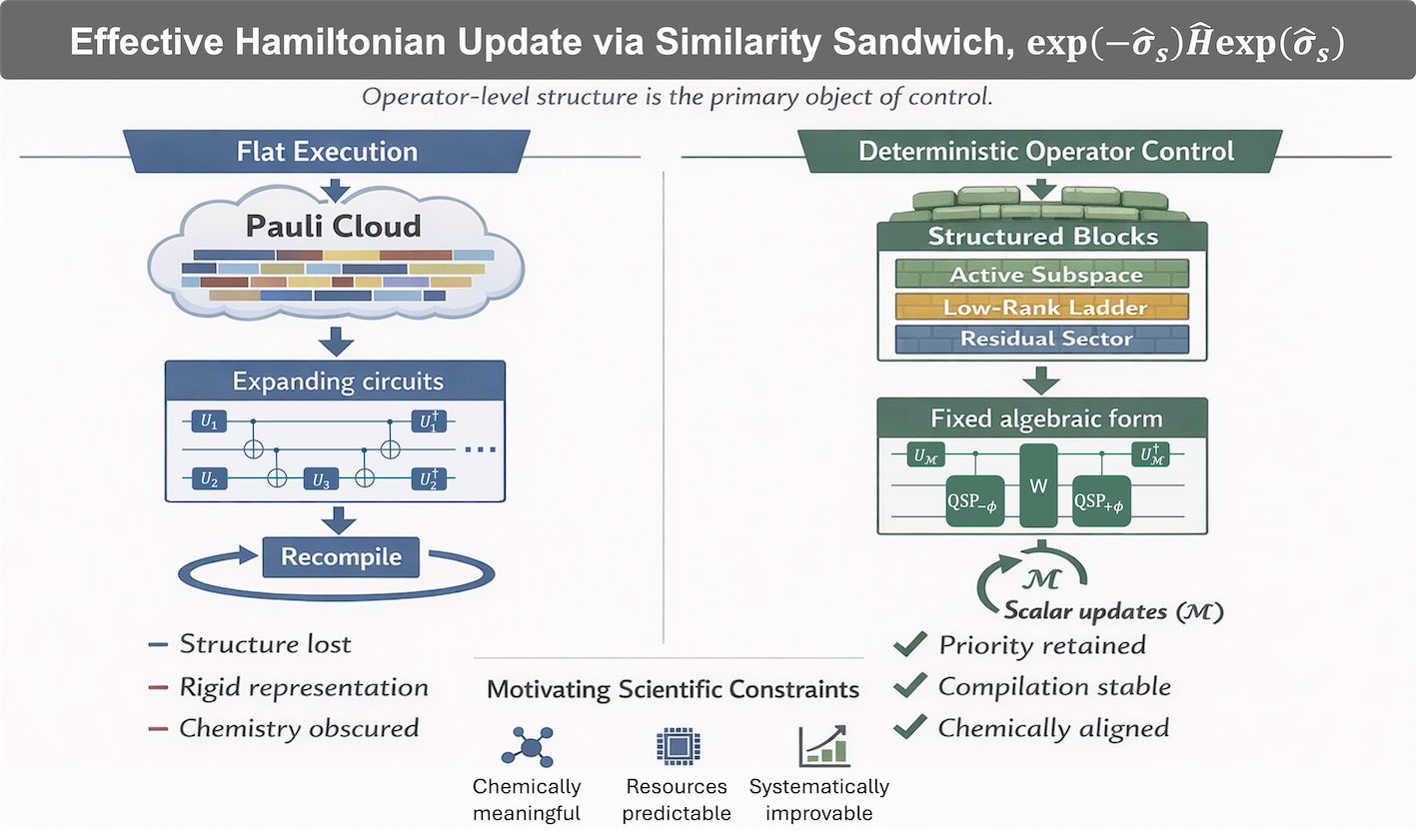}
\vspace{-0.2cm}
\caption{Conceptual motivation for operator-level structure control. Conventional workflows flatten second-quantized operators ($\hat H$, $\hat T$, $\hat \sigma$, \ldots) into unstructured Pauli expansions, leading to expanding circuit constructions and repeated recompilation during iterative optimization or instance updates. In this flat execution paradigm, excitation hierarchy and chemically meaningful subspaces are obscured at the circuit level. In contrast, the \texttt{COMPOSER} architecture treats operator structure as the primary object of control. Nested low-rank factorizations yield structured blocks (e.g., active subspaces, low-rank ladders, residual sectors) that are embedded within a fixed algebraic topology. Instance dependence enters only through scalar parameter updates, enabling stable compilation and predictable resource scaling. This shift from gate-level optimization to operator-level control, implemented as a fixed-topology block-encoding architecture with streamed parameters, is guided by three motivating scientific constraints: preserving chemical meaning, ensuring predictable resource growth, and enabling systematic hierarchy refinement.}
\label{fig:motivation}
\end{figure*}
%%%%%%%%%%%%%%%%%%%%%%%%%%%%%%%%%%%%%%%%

\section{Rank-One Representations of Quantum-Chemistry Operators}
\label{sec:rank_1_rep}

Throughout this section we use standard second-quantized notation for fermionic creation and annihilation operators acting on an orthonormal spin-orbital basis of size $M_{\mathrm{so}}$. The system register consists of $n = M_{\mathrm{so}}$ qubits, while ancillary qubits are introduced only to realize block encodings, selector registers, and quantum signal processing (QSP) primitives. Rank-one operators are labeled by an index $s$, and collections of such operators are assembled using binary-addressed selector registers. A classical \emph{mask} $\mathcal M^{(m)}$ specifies subsets of rank-one terms retained in truncated generators or effective Hamiltonians. Symbols used throughout the manuscript are summarized in Table~\ref{tab:notation}.

\begin{table*}
\centering
\begin{tabular}{ll}
\hline
Symbol & Meaning \\ \hline 
$M_{\mathrm{so}}$ & \# spin-orbitals in the chosen basis   \\
$n$               & \# qubits in the \emph{system} register ($n=M_{\mathrm{so}}$) \\
$N_O, N_V$        & \# occupied and virtual spin orbitals in the chosen active space \\
$a^{\dagger},a$ & Fermionic creation/annihilation operators    \\
$\hat a^{\dagger},\hat a$ & Linear combination of Fermionic creation/annihilation operators    \\
$\lambda$         & Complex prefactor of a rank-one operator \\
$\hat L_s$ & Rank-one operator \\
$\hat{A}_s=\hat L_s-\hat L_s^\dagger$ & anti-Hermitian rank-one term \\
$\hat{\mathbb A}_s=i(\hat L_s-\hat L_s^\dagger)$ & Hermitian rank-one term \\
$\omega_s,\Omega_s\in\mathbb R$ & Real coefficient of a rank-one operator \\
$\ell_H$            & \# rank-one operators in Hamiltonian factorization         \\
$\ell_\sigma$       & \# rank-one terms in anti-Hermitian operator $\hat \sigma$ \\
$\mathcal M^{(m)}$ & Classical \emph{mask} selecting labels $s$  \\
$\alpha$         & Normalization constant of a block encoding  \\
$\mathfrak{a}=\max(\lceil\log_2 \ell_H\rceil,\lceil\log_2 \ell_\sigma\rceil)$ & Selector-register width \\
$d$              & Quantum signal processing (QSP) polynomial degree \\
$U_{\texttt{prep}},W_{\texttt{sel}}$ & \texttt{PREP} loader and \texttt{SELECT} cascade  \\ \hline
\end{tabular}
\caption{Summary of notation and key algebraic symbols used in the \texttt{COMPOSER} framework and rank-one block-encoding constructions.{\color{black}
Unless otherwise stated we absorb phases into single-particle vectors so that Hamiltonian and generator weights ($\Omega_s,\omega_s$) may be taken real.}
}\label{tab:notation}
\end{table*}

\subsection{Rank-one operators}
\label{subsec:rank1}

We begin by defining a common notion of \emph{rank-one operators} in operator space, which will serve as the fundamental building blocks for Hamiltonians, similarity generators, and their block encodings throughout this work.
{\color{black}
We assume a standard second-quantized fermion-to-qubit encoding with one qubit per spin-orbital (e.g., Jordan--Wigner or parity mappings), so $n=M_{\mathrm{so}}$. The deterministic number-conserving adaptors used later are built from Givens-rotation primitives and fermionic swap-network routing \cite{Kivlichan2018FermionicSwap,Babbush2018LowDepth}.}

\begin{definition}[Rank-one bilinear operator]
Let $u,v\in\mathbb{C}^{M_{\mathrm{so}}}$ be normalized single-particle coefficient vectors, 
{\color{black}
\begin{align}
u&=(u_1,u_2,\dots,u_{n})^{\mathsf T},~~
v=(v_1,v_2,\dots,v_{n})^{\mathsf T},
\end{align}
with $\sum_{p=1}^{n}|u_p|^2=1$ and $\sum_{q=1}^{n}|v_q|^2=1$.}  $\lambda\in\mathbb{C}$ a scalar. Define the projected creation and annihilation operators
\begin{align}
\hat a_{u}^{\dagger}=\sum_{p=1}^{M_{\mathrm{so}}}u_{p}\,a_{p}^{\dagger}, ~~~
\hat a_{v}          =\sum_{q=1}^{M_{\mathrm{so}}}v_{q}^{*}\,a_{q}, \label{eq:rank1_bilinear}
\end{align}
with the subscript $u$ a \emph{label for the whole vector}, not an orbital index, and set
\begin{align}
\hat L^{(I)} ~=~ \lambda\,\hat a_{u}^{\dagger}\hat a_{v}. \label{eq:L_one_body}
\end{align}
The matrix elements $L_{pq}=\lambda u_{p}v_{q}^{*}$ form an outer product of two length-$M_{\mathrm{so}}$ vectors. The term ``rank-one'' here refers to the outer-product structure of the coefficient tensor in the single-particle operator basis, not to the rank of the induced operator on the many-body Hilbert space.
\end{definition}

{\color{black}\noindent\textbf{Remark (canonical mode transformations).}
The linear redefinition of fermionic modes in Eq.~\eqref{eq:rank1_bilinear} is a particle-number-conserving special case of a general fermionic canonical (Bogoliubov--Valatin) transformation~\cite{Bogoliubov1958NewMethodSuperconductivity,Valatin1958CommentsSuperconductivity}, obtained by restricting to transformations that do not mix creation and annihilation operators.}

\begin{definition}[Rank-one pair-excitation operator]
Let $U\in\mathbb{C}^{N_V\times N_V}$ and $V\in\mathbb{C}^{N_O\times N_O}$ be antisymmetric tensors on the virtual and occupied subspaces, respectively. Define the pair creation and annihilation operators
\begin{align}
\hat B^\dagger[U]=\sum_{a<b} U_{ab}\, a^\dagger_a a^\dagger_b,~~
\hat B[V]=\sum_{i<j} V_{ij}^{*} \, a_j a_i .\label{eq:pair_ops}
\end{align}
For $\lambda\in\mathbb{C}$, define
\begin{align}
\hat L^{(II)} ~\equiv~ \lambda\, \hat B^\dagger[U]\, \hat B[V].
\end{align}
Although $\hat L$ is quartic in fermionic operators, it is rank-one in the composite pair-excitation index $(ab,ij)$, since its coefficient tensor factorizes as $L_{ab,ij}=\lambda\,U_{ab}V_{ij}$.
\end{definition}

\noindent\textbf{Remark.} If $x,y\in\mathbb{C}^{N_V}$ and $r,s\in\mathbb{C}^{N_O}$, one may take $U_{ab}=(x_a y_b-x_b y_a)/\sqrt{2}$ and $V_{ij}=(r_i s_j-r_j s_i)/\sqrt{2}$, yielding $\hat B^\dagger[U]=a^\dagger(x)~ a^\dagger(y)$ and $\hat B[V]=a(s)~a(r)$, where
\begin{align}
    a^\dagger(x) &= \sum_{a}x_a a_a^\dagger,~~a^\dagger(y) = \sum_{b}y_b a_b^\dagger, \\
    a(s) &= \sum_{j} s_j^{*} a_j,~~a(r) = \sum_{i} r_i^{*} a_i.
\end{align}
{\color{black}
Equivalently, $U$ and $V$ may be viewed as vectors in the antisymmetric pair spaces $\mathbb{C}^{N_V\times N_V}$ and $\mathbb{C}^{N_O\times N_O}$.}
This form directly maps to the two-electron deterministic ladders constructed in Sec.~\ref{subsec:prep_two_electron} and also discussed in Refs.~\citenum{Kivlichan2018FermionicSwap,Babbush2018LowDepth,Babbush2018EncodingElectronicSpectra}. {\color{black} 
Without loss of generality one may normalize $u$ and $v$ and absorb their norms into $\lambda$, which is convenient for deterministic state preparation.}

\begin{definition}[Projected quadratic rank-one operator]
This construction captures Jastrow-like~\cite{Jastrow1955} quadratic occupation operators in a rank-one form. 
Let $\{u^{(r)}\}_{r=1}^{R}\subset\mathbb{C}^{M_{\mathrm{so}}}$ be a collection of single-particle coefficient vectors and define
$\hat a^\dagger(u^{(r)})=\sum_{p=1}^{M_{\mathrm{so}}}u^{(r)}_{p}\,a_p^\dagger$ and $\hat n_r=\hat a^\dagger(u^{(r)})\hat a(u^{(r)})$.
For coefficients $\mathbf C=(C_1,\dots,C_{R})\in\mathbb{C}^{R}$ define
\begin{align}
\hat O
&= \sum_{r=1}^{R}C_r\,\hat n_r, \\
\hat L^{(III)}
&= \hat O\,\hat O^\dagger.
\end{align}
Although $\hat L^{(III)}$ expands into many quartic monomials, its coefficient matrix over the $(r,r')$ index pair is the outer product $C_r C_{r'}^*$, hence rank-one in that index space.
\end{definition}

\noindent {\color{black} In what follows we write $\hat L$ for any rank-1 ladder, omitting the superscripts $(I)-(III)$; the appropriate type is determined by the accompanying index set (one-electron vs pair vs number-conserving) and the surrounding discussion.}

\subsection{Rank-one representation of the molecular Hamiltonian}
\label{subsec:NF}

{\color{black}
This subsection shows that the molecular Hamiltonian can be organized as a linear combination of rank-one operator ladders whose \emph{count} is empirically near-linear in $M_{\mathrm{so}}$ at fixed integral-factorization thresholds.}

The electronic Hamiltonian in second quantization is
\begin{align}
  \hat H ~=~
    \sum_{pq}h_{pq}\,a^{\dagger}_{p}a_{q}
    ~+~\tfrac12\sum_{pqrs}\langle pq|rs \rangle\,a^{\dagger}_{p}a^\dagger_{q}a_{s}a_{r}
    ~+~E_{\mathrm{nn}}.
\end{align}
We use the physicist’s (non-antisymmetrized) two-electron integrals $\langle pq|rs \rangle$, since fermionic antisymmetry is enforced by operator ordering and this form is directly amenable to Cholesky factorization. The nuclear-repulsion constant $E_{\mathrm{nn}}$ contributes only an overall energy shift and may therefore be dropped or tracked separately.

Applying a pivoted Cholesky decomposition to the two-electron tensor yields \cite{Whitten1973RI,Vahtras1993RI,BeebeLinderberg1977,Aquilante2007Cholesky,RoeeggenJohansen2008Cholesky,2020DFReview}
\begin{align}
  \langle pq|rs \rangle ~\approx~\sum_{\mu=1}^{K}L_{pr}^{\mu}L_{qs}^{\mu},
  \label{eq:cholesky}
\end{align}
where each $L^{\mu}$ is an $M_{\mathrm{so}}\times M_{\mathrm{so}}$ symmetric matrix and $K$ is chosen according to a prescribed threshold. Related integral-compression approaches such as density fitting / resolution-of-the-identity and tensor hypercontraction provide alternative low-rank representations of the same two-electron tensor and have been widely used in classical electronic-structure theory \cite{Dunlap2000RobustFitting,Hohenstein2012THC1,Hohenstein2012THC2}. After the standard mean-field shift
\begin{align}
  \tilde{h}_{pq} = h_{pq} - \tfrac12 \sum_s \langle pq|ss \rangle,
\end{align}
both $\tilde{h}$ and each Cholesky factor $L^{\mu}$ may be diagonalized in their respective single-particle bases. Note that the diagonalizing rotation is generally $\mu$-dependent, so the operators $\hat n_{\mu\xi}$ below live in different one-particle bases for different $\mu$.

Rewriting the Hamiltonian in the resulting rotated orbitals yields~\cite{Berry2019ArbitraryBasis,motta2020lowrank,vonBurg2022THC,2020DFReview}
\begin{align}
  \hat H
  &=\sum_{\eta=1}^{R_{1}}\kappa_{\eta}\,\hat a_{\eta}^{\dagger}\hat a_{\eta}
    +\tfrac12\sum_{\mu=1}^{K}
      \Bigl(\sum_{\xi=1}^{R_\mu} \lambda_{\xi}^{\mu}\hat n_{\mu\xi}\Bigr)^2, \label{eq2:quadratic_rank_1}
\end{align}
{\color{black}
Here $R_1\le M_{\mathrm{so}}$ denotes the number of retained eigenmodes of $\tilde h$ after optional truncation; without truncation $R_1=M_{\mathrm{so}}$. Similarly, $R_\mu$ is determined by the retained spectral rank of the $\mu$th factor after thresholding.}
$\hat a_{\eta}^{\dagger}$ and $\hat n_{\mu\xi}$ denote creation and occupation operators in the corresponding rotated bases. Importantly, both the one-electron and two-electron contributions are expressed using the same rank-one ladder structure introduced in Sec.~\ref{subsec:rank1}.
{\color{black}
In particular, for each Cholesky channel $\mu$ we define the diagonal one-body operator
\begin{align}
\hat O_\mu := \sum_{\xi=1}^{R_\mu}\lambda^{(\mu)}_\xi\,\hat n_{\mu\xi},
\end{align}
so that the channel contribution is the {\em projected quadratic rank-one} operator
\begin{align}
\hat L_\mu := \hat O_\mu \hat O_\mu^\dagger = \hat O_\mu^2,
\end{align}
which is precisely of the form in {\bf Definition~3} with coefficient vector $C=\lambda^{(\mu)}$.}

Collecting terms, the Hamiltonian may therefore be written compactly as
\begin{align}
  \hat H = \sum_{s=1}^{\ell_H}\Omega_{s}\hat L_{s},
  \label{eq:H_rank_1}
\end{align}
where each $\hat L_s$ is a rank-one operator and
\begin{align}
  \ell_H = R_1 + K.
  \label{eq:Ham_rank1}
\end{align}
For a fixed Cholesky threshold, empirical studies typically find that $K$ (and hence $\ell_H$) grows approximately linearly with $M_{\mathrm{so}}$ for a wide range of molecular systems \cite{Berry2019ArbitraryBasis,Rocca2024SCDF}, which is the key structural property that enables fixed-topology block encodings in the \texttt{COMPOSER} architecture.

\subsection{Rank-one representation of the cluster generator}
\label{subsec:cluster}

We now show that the anti-Hermitian coupled-cluster doubles generator admits a rank-one decomposition compatible with the operator ladders introduced in Sec.~\ref{subsec:rank1}. Throughout this subsection we adopt the convention that occupied and virtual indices are explicitly antisymmetrized.

We focus on the doubles excitation operator,
\begin{align}
\hat T_2
= \sum_{a<b}\sum_{i<j} t_{ab,ij}\, a_a^\dagger a_b^\dagger a_j a_i,
\end{align}
where $(i,j)$ label occupied spin-orbitals and $(a,b)$ label virtual spin-orbitals in an orthonormal spin-orbital basis. Singles contributions are inherently rank-one and may be treated analogously; higher-order excitations can be incorporated using the same principles.

{\color{black}
To expose low-rank structure while respecting fermionic antisymmetry, we view $t_{ab,ij}$ as a matrix on antisymmetric pair spaces by introducing composite indices $ab\equiv(a<b)$ and $ij\equiv(i<j)$, so that $t_{ab,ij}$ is a $\binom{N_V}{2}\times \binom{N_O}{2}$ matrix. Equivalently, for each $\mu$ the factors below may be unpacked as skew-symmetric tensors $\mathbb{M}_{ab}^{(\mu)}=-\mathbb{M}_{ba}^{(\mu)}$ and $\mathbb{N}_{ij}^{(\mu)}=-\mathbb{N}_{ji}^{(\mu)}$.}
We then perform a singular-value decomposition \cite{Parrish2019RankReducedCC,Hohenstein2022RankReducedCC,Peruzzo2014VQE,McClean2016TheoryVQE,OMalley2016ScalableQChem,McClean2017QSE},
\begin{align}
t_{ab,ij} = \sum_{\mu=1}^{R_{\mu}} \epsilon_\mu \,
\mathbb{M}_{ab}^{(\mu)} \,
\mathbb{N}_{ij}^{(\mu)},
\label{eq:T2_svd}
\end{align}
where the upper limit $R_{\mu}$ is determined by a prescribed truncation threshold for $T_2$.

{\color{black}
Each singular component is then further expressed as a sum of \emph{rank-one antisymmetric pair factors} (wedge products) within its respective subspace. Concretely, we approximate}
\begin{align}
\mathbb{M}_{ab}^{(\mu)} &\approx \sum_{\kappa=1}^{R_\kappa}
\epsilon_\kappa^{(\mu)} \, U_{ab}^{(\mu\kappa)}, \\
\mathbb{N}_{ij}^{(\mu)} &\approx \sum_{\eta=1}^{R_\eta}
\epsilon_\eta^{(\mu)} V_{ij}^{(\mu\eta)},
\end{align}
where
\begin{align}
U_{ab}^{(\mu\kappa)}=\frac{x_a^{(\mu\kappa)}y_b^{(\mu\kappa)}-x_b^{(\mu\kappa)}y_a^{(\mu\kappa)}}{\sqrt{2}}, \\
\qquad
V_{ij}^{(\mu\eta)}=\frac{r_i^{(\mu\eta)}s_j^{(\mu\eta)}-r_j^{(\mu\eta)}s_i^{(\mu\eta)}}{\sqrt{2}},
\end{align}
{\color{black}
where the antisymmetry $U_{ab}^{(\mu\kappa)}=-U_{ba}^{(\mu\kappa)}$ and $V_{ij}^{(\mu\eta)}=-V_{ji}^{(\mu\eta)}$ is explicit. (The factor of $1/\sqrt{2}$ is a convention; it may be absorbed into the vectors or coefficients depending on the chosen normalization of the pair basis.)}
This yields the explicit factorization
\begin{align}
{\color{black}
t_{ab,ij}
\approx
\sum_{\mu,\kappa,\eta}
\epsilon_\mu \epsilon_\kappa^{(\mu)} \epsilon_\eta^{(\mu)} \,
U_{ab}^{(\mu\kappa)} \,
V_{ij}^{(\mu\eta)} .}
\end{align}

{\color{black}
Employing the definition of the pair creation/annihilation operators in Eq.~\eqref{eq:pair_ops}, we can write the rank-one pair-excitation ladder as
\begin{align}
\hat L_{\mu,\kappa,\eta}\equiv \hat B^\dagger[U^{(\mu\kappa)}]\hat B[V^{(\mu\eta)}].
\end{align}}
The doubles operator therefore admits the rank-one expansion
\begin{align}
\hat T_2
\approx \sum_{\mu,\kappa,\eta}
\omega_{\mu,\kappa,\eta}\,
\hat L_{\mu,\kappa,\eta}
\end{align}
with $\omega_{\mu,\kappa,\eta}:=\epsilon_\mu \epsilon_\kappa^{(\mu)} \epsilon_\eta^{(\mu)}$.
The corresponding anti-Hermitian generator is
\begin{align}
\hat \sigma_2
&= \hat T_2 - \hat T_2^\dagger \notag \\
&\approx \sum_{\mu,\kappa,\eta}
\omega_{\mu,\kappa,\eta}
\left(
\hat L_{\mu,\kappa,\eta}
-
\hat L_{\mu,\kappa,\eta}^\dagger
\right),
\label{eq:sigma_rank_1}
\end{align}
a structured sum of rank-one ladders. {\color{black}
Phases may be absorbed into the orbital vectors so that $\omega_{\mu,\kappa,\eta}\in\mathbb{R}$, making each $\omega_{\mu,\kappa,\eta}(\hat L-\hat L^\dagger)$ manifestly anti-Hermitian.}

This representation makes explicit that the coupled-cluster generator is built from the same algebraic primitives as the Hamiltonian decomposition in Sec.~\ref{subsec:NF}. As a result, Hamiltonians and similarity generators can be treated on equal footing within the \texttt{COMPOSER} block-encoding architecture.

%%%%%%%%%%%%%%%%%%%%%%

\section{Deterministic preparation of one- and two-electron states}\label{sec:Det_prep}

In this section we construct deterministic ladder circuits that realize the rank-one operators introduced in Sec.~\ref{sec:rank_1_rep}. {\color{black}The core objects are number-conserving ladder unitaries $U_u$ built from Givens (and pair--Givens) rotations, optional routing swaps, and diagonal phase shifts.} These ladders serve a dual role: when {\color{black}preceded by a simple pivot injection} and applied to the vacuum, they prepare one- and two-electron states with prescribed amplitudes; when applied {\color{black}without injection}, they act as basis-rotation primitives inside rank-one block-encoding adaptors. The same fixed ladder topology is reused throughout the \texttt{COMPOSER} architecture.

\subsection{Deterministic preparation of the one-electron state}
\label{subsec:prep_u_state}

The target state in the one-electron sector is
\begin{align}
\lvert u\rangle
  = \hat a_u^\dagger \lvert 0^n\rangle
  = \sum_{p=1}^{n} u_p\, \lvert p\rangle ,
\label{eq:one_e_state}
\end{align}
where $\lvert p\rangle = a_p^\dagger \lvert 0^n\rangle$ with $p\in \{1,\cdots, n\}$ denotes the Slater determinant with the $p$th spin-orbital occupied.

{\color{black}
We construct a deterministic preparation circuit by composing (i) a single-electron pivot injection and (ii) a number-conserving ladder unitary. Concretely, fix a pivot orbital $r$ (chosen once to define a fixed ladder topology) and define $\lvert r\rangle = X_r\lvert 0^n\rangle$. We then build a number-conserving unitary $U_u$ from two-mode Givens rotations and diagonal phase shifts such that $U_u\lvert r\rangle=\lvert u\rangle$, so that the full preparation circuit is $U^{(u)}_{\texttt{prep}} = U_u X_r$ \cite{Kivlichan2018FermionicSwap,Babbush2018LowDepth,Babbush2018EncodingElectronicSpectra}.}
{\color{black}
For numerical stability one may choose $r$ so that $|u_r|$ is not anomalously small (e.g., $r=\arg\max_p |u_p|$), but this is not required for correctness; it only affects the conditioning of the classical angle recursion.} 

For each $p\neq r$, the number-conserving Givens rotation is defined as
\begin{align}
G_{pr}(\theta)
= \exp\left[
\theta\left(a_p^\dagger a_r - a_r^\dagger a_p\right)
\right],
\label{eq:Givens_def}
\end{align}
which acts as an $\mathrm{SO}(2)$ rotation on $\mathrm{span}\{\lvert r\rangle,\lvert p\rangle\}$. A fixed ordering of the non-pivot orbitals defines a ladder that sequentially transfers amplitude out of the pivot mode; the required rotation angles and residual single-qubit phases are computed classically from $u$. Here the data vector $u$ (and similarly $u_{pq}$ in the two-electron ladder) is supplied by the chosen operator factorization/parameterization—e.g., integral-derived low-rank factors for $\hat H$ or solver-seeded factors for $\hat \sigma$---and is not restricted to coupled-cluster doubles. Explicit recursions and an explicit product form are collected in Appendix~\ref{app:ladders}. The resulting circuit prepares $\lvert u\rangle$ deterministically; its inverse is obtained by reversing the gate order and negating all angles. {\color{black}
The one-electron ladder uses $(n-1)$ two-mode rotations (plus phases and optional routing swaps), so its gate count scales as $\mathcal{O}(n)$ for fixed connectivity assumptions.}

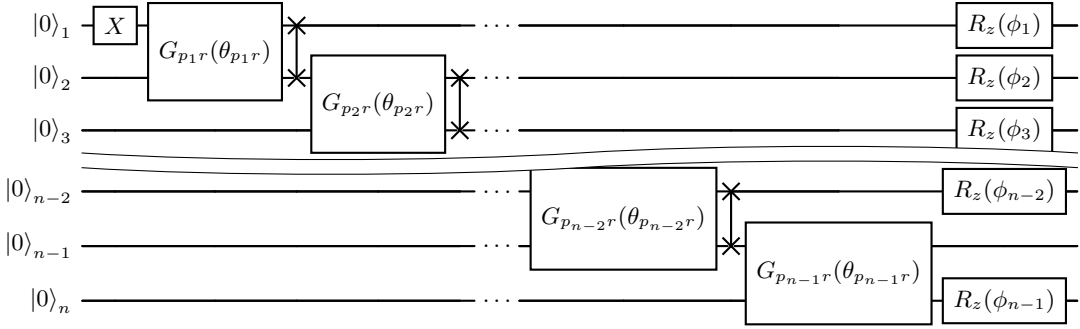
\begin{figure*}
\centering
    \begin{quantikz}[row sep=0.1cm, column sep=0.15cm, font=\small]
        \lstick{$\ket{0}_1$} 
        & \gate{X} 
        & \gate[wires=2]{G_{p_1r}(\theta_{p_1r})}  
        & \swap{1} 
        & \qw
        & \qw
        & ~\ldots~
        & \qw 
        & \qw 
        & \qw 
        & \gate{R_z(\phi_1)} 
        & \qw \\
        \lstick{$\ket{0}_2$} 
        & \qw 
        & \qw 
        & \swap{0} 
        & \gate[wires=2]{G_{p_2r}(\theta_{p_2r})} 
        & \swap{1}
        & ~\ldots~
        & \qw 
        & \qw 
        & \qw 
        & \gate{R_z(\phi_2)} 
        & \qw \\
        \lstick{$\ket{0}_3$} 
        & \qw 
        & \qw 
        & \qw 
        & \qw
        & \swap{0}
        & ~\ldots~
        & \qw 
        & \qw 
        & \qw 
        & \gate{R_z(\phi_3)} 
        & \qw \\
        \wave &&&&&&&&&&& \\
        \lstick{$\ket{0}_{n-2}$} 
        & \qw 
        & \qw 
        & \qw 
        & \qw 
        & \qw 
        & ~\ldots~ 
        & \gate[wires=2]{G_{p_{n-2}r}(\theta_{p_{n-2}r})}
        & \swap{1} 
        & \qw 
        & \gate{R_z(\phi_{n-2})} 
        & \qw \\
        \lstick{$\ket{0}_{n-1}$} 
        & \qw 
        & \qw 
        & \qw 
        & \qw 
        & \qw 
        & ~\ldots~ 
        & \qw
        & \swap{0} 
        & \gate[wires=2]{G_{p_{n-1}r}(\theta_{p_{n-1}r})} & \qw 
        & \qw \\
        \lstick{$\ket{0}_n$} 
        & \qw 
        & \qw 
        & \qw 
        & \qw 
        & \qw 
        & ~\ldots~
        & \qw
        & \qw 
        & \qw 
        & \gate{R_z(\phi_{n-1})} 
        & \qw 
    \end{quantikz}
\caption{Deterministic ladder circuit for one-electron state preparation. The ladder topology is fixed; only the angles $\{\theta_{pr}\}$ and phases $\{\phi_p\}$ depend on the data vector $u$ (see Appendix~\ref{app:ladders}). {\color{black}The interleaved swaps indicate a connectivity-aware routing pattern; when the ladder is used as a basis rotation on general $N$-electron states, these swaps should be interpreted as fermionic swap-network primitives (or an equivalent orbital-routing compilation)~\cite{Kivlichan2018FermionicSwap}.}}
\label{fig:one_electron_prep_corrected}
\end{figure*}

\noindent\textbf{Remark (number-conserving form).}
{\color{black}
By construction $U_u$ is strictly number-conserving (it preserves each fixed-$N$ sector); the preparation form differs only by the initial pivot injection $X_r$.}
This is the form reused inside rank-one block-encoding adaptors, independent of whether the system register holds the vacuum or an arbitrary $N$-electron state.

\subsection{Deterministic preparation of a two-electron state}
\label{subsec:prep_two_electron}

The two-electron ladder generalizes Sec.~\ref{subsec:prep_u_state} to antisymmetric excitation pairs~\cite{Kivlichan2018FermionicSwap,Babbush2018LowDepth,Babbush2018EncodingElectronicSpectra}. Let $u_{pq}=-u_{qp}$ be a normalized antisymmetric coefficient tensor, $\sum_{p<q}|u_{pq}|^2=1$. The target state in the two-electron sector is
\begin{align}
\lvert u\rangle
=
\sum_{p<q} u_{pq}\,\lvert pq\rangle,
~~
\lvert pq\rangle
=
a_p^\dagger a_q^\dagger \lvert 0^n\rangle .
\end{align}

{\color{black}
Fix a pivot pair $(r,s)$ with $r<s$ (chosen once to define a fixed ladder topology) and inject two particles,}
\begin{align}
X_r X_s \lvert 0^n\rangle = \lvert rs\rangle .
\end{align}
{\color{black}
Here, as in the one-electron case, choosing $(r,s)$ so that $|u_{rs}|$ is not anomalously small improves the conditioning of the classical angle/phase recursion but is not required for correctness.}
Amplitude is then redistributed from the pivot pair using phased pair--Givens rotations. For each unordered pair $\{p,q\}\neq\{r,s\}$, define
\begin{align}
G_{pq,rs}(\theta,\phi)
=
\exp\Big[
\theta\big(
e^{i\phi}\,a_p^\dagger a_q^\dagger a_s a_r
-
e^{-i\phi}\,a_r^\dagger a_s^\dagger a_q a_p
\big)
\Big],
\label{eq:pair_Givens_phased}
\end{align}
which preserves particle number and rotates $\mathrm{span}\{\lvert rs\rangle,\lvert pq\rangle\}$. As in the one-electron case, a fixed ordering of non-pivot pairs defines a deterministic ladder; angles and phases are computed classically from $\{u_{pq}\}$ (Appendix~\ref{app:ladders}). {\color{black}
Equivalently, define the number-conserving ladder $U_u$ such that $U_u\lvert rs\rangle=\lvert u\rangle$, so the full preparation circuit is $U^{(u)}_{\texttt{prep}}=U_u X_r X_s$.}
{\color{black}
The two-electron ladder contains one pair rotation per non-pivot pair (up to ordering/truncation), i.e., $\binom{n}{2}-1$ rotations in the worst case, so its size scales as $\mathcal{O}(n^2)$ before sparsity/truncation and connectivity-aware compilation.}\\

\noindent\textbf{Remark (number-conserving form).}
Omitting the initial excitation $X_r X_s$ yields a strictly number-conserving ladder unitary $U_u$ that acts as a basis rotation within the two-electron subspace. This number-conserving form is the one used inside pair-excitation adaptors in \texttt{COMPOSER}.

\medskip
\noindent\textbf{Preparation vs.\ number conservation.}
Throughout the remainder of this work, $U_u$ denotes the number-conserving ladder form. {\color{black}
When needed we disambiguate by writing $U^{(1)}_u$ (one-electron ladder) and $U^{(2)}_u$ (two-electron ladder); otherwise the intended sector is clear from context.} The preparation form from the vacuum differs only by the initial pivot injection (one $X$ gate for the one-electron ladder; two $X$ gates for the two-electron ladder). For completeness, we collect explicit product decompositions and a short discussion of these two realizations, as well as connectivity-aware routing, Givens--SWAP network constructions, and depth scaling for the two-electron ladders, in Appendix~\ref{app:ladders}.

%%%%%%%%%%%%%%%%%%%%%%%%%%%

\section{Block encoding of rank-one operators}\label{sec:block_encoding}

\texttt{COMPOSER} is formulated on an $n$-qubit second-quantized register, but its oracles are used only on a fixed particle-number sector $\mathcal H_{N_e}\subset(\mathbb C^2)^{\otimes n}$ (and, when indicated, on a user-chosen model subspace $\mathrm{Ran}(P^{(m)})\subseteq \mathcal H_{N_e}$). All ladder primitives and adaptor branches used in COMPOSER are strictly number conserving (Sec.~\ref{sec:Det_prep} and Appendix~\ref{app:ladders}), so $\mathcal H_{N_e}$ is an invariant subspace of every \texttt{SELECT} branch and of the assembled \texttt{PREP-SELECT-PREP}${}^\dagger$ circuits.

Accordingly, throughout this section, statements of the form $(\langle 0^t|\otimes I)\,W_{\hat O}\,(|0^t\rangle\otimes I)\approx \hat O/\alpha$ are to be understood as holding in operator norm \emph{after restriction to the working subspace}. That is, for $\Pi_{N_e}$ the orthogonal projector onto $\mathcal H_{N_e}$, we require
\[
\bigl\|\Pi_{N_e}\bigl[(\langle 0^t|\otimes I)\,W\,(|0^t\rangle\otimes I)-\hat O/\alpha\bigr]\Pi_{N_e}\bigr\|\le \epsilon.
\]
Outside the working subspace (e.g., on other Hamming-weight sectors), the action of the block-encoding unitary is unconstrained and irrelevant to \texttt{COMPOSER}'s use cases. For readability we suppress the explicit $\Pi_{N_e}$ sandwiches in most equations below.

%%%%%%%%%%%%%%%%%%%%%%%%%%%%%
\subsection{Encoding rank-one operators}\label{subsubsec:uv_block}

\begin{lemma}[Single-ancilla dyad block encoding in the one-excitation register]\label{lem:rank1_block}
{\color{black}
Throughout this lemma we treat $\hat a_u^\dagger \hat a_v$ as an operator restricted to the embedded one-electron subspace $\mathcal{H}_{N=1}\subset(\mathbb{C}^2)^{\otimes n}$, on which it is exactly the dyad $\ket{u}\bra{v}$.}
Let $\hat L_s=\lambda_s\,\hat a_u^\dagger \hat a_v$ be a bilinear rank-one operator (Eq.~\eqref{eq:L_one_body}) with $u,v\in\mathbb C^{n}$ and $\lambda_s\in\mathbb C$. 
{\color{black}In this lemma we consider the induced action on the \emph{single-excitation} subspace $\mathcal H_{1}\subset(\mathbb C^2)^{\otimes n}$, where $\hat a_u^\dagger \hat a_v$ acts exactly as the dyad $\ket{u}\bra{v}$.}
{\color{black}Without loss of generality, phases may be absorbed into $u$ and $v$ so that $\lambda_s\in\mathbb R_{\ge 0}$.}
For any error $\epsilon>0$, there exists a unitary $W_s$ acting on the system and one ancilla qubit such that
\begin{align}
\| (\bra{0}\otimes I_S)\,W_s\,(\ket{0}\otimes I_S)
- \hat L_s/\alpha \| < \epsilon , \label{eq:be_rank_1}
\end{align}
with normalization $\alpha\ge|\lambda_s|$.
\end{lemma}

\noindent\textbf{Proof.} See Appendix~\ref{app:rank1_proofs}.\\

\noindent {\color{black}\textbf{Remark 1.} In the many-electron use cases of \texttt{COMPOSER}, the same ladder primitives $U_u,U_v$ appear as number-conserving orbital-rotation subroutines inside the full \texttt{SELECT} table; the dyad viewpoint is used here only to give a compact correctness proof of the rank-one adaptor interface.}\\

\noindent \textbf{Remark 2.} Based on Lemma~\ref{lem:rank1_block}, we can further block encode the the pair-excitation ladder $\hat B^\dagger[U]\hat B[V]$. We denote $\hat L_s = \lambda_s \hat B^\dagger[U]\hat B[V]$ as a rank-one pair-excitation operator with antisymmetric pair tensors $U$, $V,$ and $\lambda_s\in\mathbb C$. Assume $U$ and $V$ are normalized in the antisymmetric pair basis, $\sum_{p<q}|U_{pq}|^2=1$, $\sum_{p<q}|V_{pq}|^2=1$, and define the corresponding normalized two-electron states $|U\rangle := \hat B^\dagger[U]|0^n\rangle$, $|V\rangle := \hat B^\dagger[V]|0^n\rangle$.
Then on $\mathcal H_{N=2}\subset(\mathbb{C}^2)^{\otimes n}$, the operator $\hat B^\dagger[U]\hat B[V]$ acts exactly as the dyad $|U\rangle\langle V|$.
Without loss of generality, $\arg(\lambda_s)$ can be absorbed into $U$ or $V$ so that $\lambda_s\in\mathbb R_{\ge 0}$. For any error $\epsilon>0$, there exists a unitary $W_s$ acting on the system register and one ancilla qubit such that Eq.~\eqref{eq:be_rank_1} holds. In the many-electron use cases of \texttt{COMPOSER}, the same two-electron ladder primitives from Sec.~\ref{subsec:prep_two_electron} appear in their number-conserving form inside the \texttt{SELECT} table; the restriction to $\mathcal H_{N=2}$ is used here only to state a compact adaptor-correctness interface, analogous to Lemma~\ref{lem:rank1_block}.

\begin{lemma}[Deterministic block encoding of a squared diagonalized Cholesky channel]
\label{lem:quad_rank1}
Let $L^\mu\in\mathbb{C}^{n\times n}$ be a (single-bar) Cholesky channel with eigendecomposition $L^\mu = U^{(\mu)}\mathrm{diag}(\lambda^{(\mu)})U^{(\mu)\dagger}$. Define rotated modes $\hat a_{\mu\xi}^{\dagger}=\sum_{p}U^{(\mu)}_{p\xi}\hat a_p^{\dagger}$ and $\hat n_{\mu\xi}=\hat a_{\mu\xi}^{\dagger}\hat a_{\mu\xi}$, and define the Hermitian operator
\begin{align}
{\color{black}
\hat O_\mu := \sum_{\xi=1}^{R_\mu}\lambda^{(\mu)}_{\xi}\,\hat n_{\mu\xi}, \qquad
\Gamma_\mu := \sum_{\xi=1}^{R_\mu}\bigl|\lambda^{(\mu)}_{\xi}\bigr|.}
\end{align}
Then $\hat O_\mu^2$ admits a \emph{deterministic} $(\Gamma_\mu^2,\,\mathfrak{a}_I+2,\,0)$ block encoding, where $\mathfrak{a}_I$ is the width of the index register (binary: $\mathfrak{a}_I=\lceil\log_2 R_\mu\rceil$; unary: $\mathfrak{a}_I=R_\mu$).
{\color{black}
Concretely, we first construct a deterministic block encoding of $\hat O_\mu/\Gamma_\mu$ using
\texttt{PREP-SELECT-PREP}$^\dagger$ over the commuting projectors $\{\hat n_{\mu\xi}\}$ (including the sign of $\lambda_\xi^{(\mu)}$ as a branch phase),
and then apply a fixed, degree-$2$ QSVT/QSP polynomial implementing $x\mapsto x^2$ to obtain $\hat O_\mu^2/\Gamma_\mu^2$.}
\end{lemma}

\noindent\textbf{Proof.} See Appendix~\ref{app:rank1_proofs}.\\

\noindent\textbf{General remark (two types of rank-one adaptors).}
Lemmas~\ref{lem:rank1_block} and~\ref{lem:quad_rank1} implement two complementary rank-one block-encoding primitives. Lemma~\ref{lem:rank1_block} encodes a \emph{bilinear} rank-one operator $\hat L=\lambda\,\hat a_u^\dagger \hat a_v$, which couples two distinct single-particle modes and is realized using deterministic one-electron state-preparation ladders. {\color{black}
In the double-factorized form, each Cholesky channel contributes $\tfrac12 \hat O_\mu^2$ with $\hat O_\mu=\sum_\xi \lambda^{(\mu)}_\xi \hat n_{\mu\xi}$; Lemma~\ref{lem:quad_rank1} provides a deterministic adaptor
for $\hat O_\mu$ and, via a fixed degree-2 QSVT step, for $\hat O_\mu^2$ without enumerating the $\mathcal O(R_\mu^2)$ cross terms explicitly.} Together, these two adaptors span the bilinear and quadratic building blocks arising from the nested factorization of molecular Hamiltonians and cluster generators, and both admit deterministic, single-signal block encodings with fixed circuit topology.

\begin{theorem}[Binary-multiplexed block encoding]
\label{thm:sum_block}
Suppose $\hat H=\sum_{s=1}^{\ell_H}\Omega_{s}\hat L_{s}$ and for each $s$ there exists a unitary $W_s$ that is an $(\alpha_s,t,\epsilon_s)$ block encoding of $\hat L_s$, acting on an $n$-qubit system register with $t$ ancillas initialized in $\ket{0^t}$. Then there exists a unitary $W$ that is an $(\alpha,\mathfrak{a}_H+t,\epsilon_{\mathrm{LCU}})$ block encoding of $\hat H$ such that
\begin{align}
\left\| \left(\bra{0^{\mathfrak{a}_H+t}}\otimes I_S\right)\,W\,\left(\ket{0^{\mathfrak{a}_H+t}}\otimes I_S \right)
-
\frac{\hat H}{\alpha} \right\| < \epsilon_{\rm LCU} ,
\end{align}
where
\[
\alpha=\sum_{s=1}^{\ell_H}|\Omega_s|\,\alpha_s,
~~
\epsilon_{\mathrm{LCU}}
= \frac1{\alpha}\sum_{s=1}^{\ell_H}|\Omega_s|\,\alpha_s\,\epsilon_s .
\]
Here $\mathfrak{a}_H=\lceil\log_2 \ell_H\rceil$ is the selector width, and the $t$ ancillas used by $W_s$ are reused across all branches~\cite{Childs2012LCU,Berry2015Taylor,Gilyen2019QSVT,Low2019Qubitization}. The corresponding binary-multiplexed \texttt{PREP-SELECT-PREP}$^\dagger$ circuit structure is shown in Figure~\ref{fig:circuit_LC}.

The circuit depth obeys
\[
\mathrm{depth}(W)
=
\mathcal O\left(
\mathrm{depth}(U_{\texttt{prep}})
+
\sum_{s=1}^{\ell_H}\mathrm{depth}(\text{c-}W_s)
\right).
\]
Here $\mathrm{depth}(\text{c-}W_s)$ denotes the depth of the selector-controlled implementation of $W_s$. For example, $W_s$ may be constructed using Lemma~\ref{lem:rank1_block} (bilinear adaptor) or Lemma~\ref{lem:quad_rank1} (diagonal Cholesky-channel adaptor), leading to representative costs $\mathrm{depth}(W_s)=\mathcal O(n)$ and $\mathrm{depth}(W_s)=\mathcal O(n+R_\mu)$, respectively, where $R_\mu$ is the number of retained rotated modes in channel $\mu$.
\end{theorem}

\noindent\textbf{Proof.} See Appendix~\ref{app:rank1_proofs}.

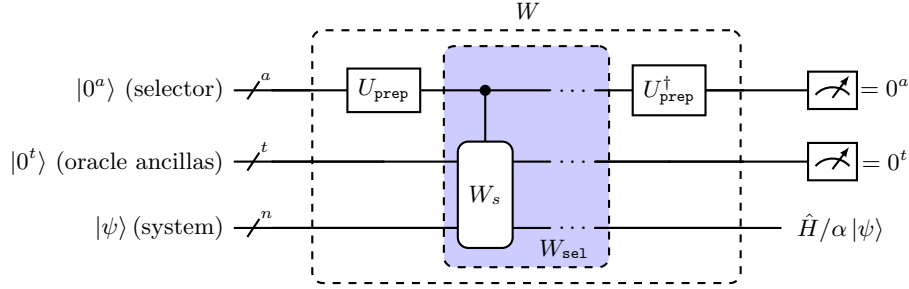
\begin{figure*}
\centering
\begin{quantikz}[row sep=0.35cm, column sep=0.5cm]
\lstick{$\ket{0^a}$ (selector)} 
& \qwbundle{a} 
& \qw 
& \gate{U_{\texttt{prep}}}\gategroup[3,steps=4,style={dashed,rounded corners,inner sep=10pt},background]{\textit{W}} 
& \ctrl{1}\gategroup[3,steps=2,style={dashed,rounded corners,fill=blue!20, inner xsep=2pt},background,label style={label position=below,anchor=north,xshift=0.5cm,yshift=0.3cm}]{$W_\texttt{sel}$} 
& ~\ldots~ 
& \gate{U_{\texttt{prep}}^\dagger} 
& \qw 
& \meter[label style={label position=below,anchor=north,xshift=0.7cm}]{=0^a} \\
\lstick{$\ket{0^t}$ (oracle ancillas)}
& \qwbundle{t}
& \qw 
& \qw                       
& \gate[wires=2,style={rounded corners}]{W_s} 
& ~\ldots~ 
& \qw 
& \qw 
& \meter[label style={label position=below,anchor=north,xshift=0.7cm}]{=0^t} \\
\lstick{$\ket{\psi}$\,(system)} 
& \qwbundle{n} 
& \qw 
& \qw 
& \qw 
& ~\ldots~
& \qw 
& \qw 
&~~\hat H/\alpha \ket{\psi}
\end{quantikz}
\caption{Binary-multiplexed \texttt{PREP-SELECT-PREP$^{\dagger}$} circuit realizing the unitary \(W\).
{\color{black}The \texttt{SELECT} stage implements the controlled table
$W_{\texttt{sel}}=\sum_{s} \ket{s}\bra{s}\otimes (e^{i\phi_s}W_s)$, where phases $\phi_s$ absorb coefficient signs/phases.}
The blue dashed fan-out applies \(W_s\) when the selector encodes \(|s\rangle\). {\color{black}The workspace register has fixed width $t=\max_s t_s$ (including any adaptor-internal index/flag/signal qubits); adaptors with smaller workspace act trivially on unused ancillas.}
The \(t\) oracle ancillas required by a single \(W_s\) are reused across branches and projected to \(|0^t\rangle\).}\label{fig:circuit_LC}
\end{figure*}

%%%%%%%%%%%%%%
\subsection{Encoding exponentials of rank-one generators}
\label{subsubsec:pqr1_exp_block}

We now describe how to implement the exponential of an anti-Hermitian generator,
\begin{align}
\hat\sigma = \sum_{s=1}^{\ell_\sigma} \omega_s \hat A_s,
~~
\hat A_s = \hat L_s - \hat L_s^\dagger, \label{eq:op_sigma}
\end{align}
where each $\hat L_s$ is a rank-one operator and $\omega_s\in\mathbb R$. Defining the associated Hermitian operators
\begin{align}
\hat{\mathbb A}_s := i\hat A_s,
\end{align}
we may rewrite
\begin{align}
e^{\hat\sigma} = e^{-i\sum_s \omega_s \hat{\mathbb A}_s}.
\label{eq:expsigma}
\end{align}
We emphasize that QSP is applied only to the block encoding of the generator $\hat\sigma$. The Hamiltonian block encoding constructed in Sec.~\ref{subsubsec:uv_block} is not exponentiated and therefore does not incur polynomial-degree overhead.

\paragraph{Block encoding of the generator.}
Using Lemma~\ref{lem:rank1_block} and Theorem~\ref{thm:sum_block}, each $\hat{\mathbb A}_s$ admits a $(2\alpha_s,1,2\epsilon_s)$ block encoding $W_{\hat{\mathbb A}_s}$, where the factor of $2$ accounts for $\hat L_s$ and $\hat L_s^\dagger$. Applying the \texttt{PREP-SELECT-PREP}$^\dagger$ construction yields a block encoding of the full generator,
\begin{align}
U_{\hat\sigma}
=
(U_{\texttt{prep}}^\dagger\otimes I)
\Bigl(
\sum_{s=1}^{\ell_\sigma} \ket{s}\bra{s}\otimes W_{\hat{\mathbb A}_s}
\Bigr)
(U_{\texttt{prep}}\otimes I),
\end{align}
with
\begin{align}
U_{\texttt{prep}}(\mathcal M^{(m)})\ket{0^{\mathfrak{a}'}}
&=
\frac{1}{\sqrt{\bar\alpha'}}
\left(
\sum_{s\in\mathcal M^{(m)}} \sqrt{2|\omega_s|\alpha_s}\,\ket{s} \right. \notag \\
&\left. +
\sqrt{\bar\alpha'-\sum_{s\in\mathcal M^{(m)}}2|\omega_s|\alpha_s}\,\ket{0}
\right),
\end{align}
{\color{black}
where $\mathfrak{a}'=\lceil\log_2 \ell_\sigma\rceil$ and the additional label $\ket{0}$ is a \emph{null branch} whose \texttt{SELECT} action is the identity. Also, here we fix a \emph{global} normalization 
\begin{align}
\bar \alpha'=\sum_{s=1}^{\ell_\sigma} 2|\omega_s|\alpha_s
\end{align}
that is independent of the mask (worst-case over the full compiled term list). Because $\omega_s\in\mathbb R$ may have either sign, we implement $\mathrm{sgn}(\omega_s)$ as a fixed phase in the corresponding \texttt{SELECT} branch (equivalently, as a $\pi$ phase on $\ket{s}$ in the \texttt{PREP} state), so that the effective linear combination reproduces $\sum_s \omega_s \hat{\mathbb A}_s$ rather than $\sum_s |\omega_s|\hat{\mathbb A}_s$.}
Projecting the ancillas yields
\begin{align}
\| (\bra{0^{\mathfrak{a}'+1}}\otimes I)\,U_{\hat\sigma}\,(\ket{0^{\mathfrak{a}'+1}}\otimes I)
-
\hat\sigma/\bar{\alpha}' \|  \le \epsilon',
\end{align}
with the block-encoding error
\begin{align}
\epsilon'=\frac{4}{\bar{\alpha}'}\sum_s |\omega_s|\alpha_s\,\epsilon_s .
\label{eq:eps'}
\end{align}

\paragraph{Exponentiation via a single QSP ladder~\cite{LowChuang2017QSP,Gilyen2019QSVT}.}
Given the block encoding $U_{\hat\sigma}$, we apply quantum signal processing to approximate the matrix function $f(x)=e^{-i\bar\alpha' x}$ on $[-1,1]$. This yields
\begin{align}
\| (\bra{0}\otimes I)\,
Q(\boldsymbol{\phi},U_{\hat\sigma})\,
(\ket{0}\otimes I)
-
e^{\hat\sigma} \| \le \epsilon''(\epsilon',\varepsilon),
\label{eq:block_encode_exp_sigma}
\end{align}
where $\epsilon''$ depends linearly on $\epsilon'$ (the block-encoding error) and $\varepsilon$ (the polynomial approximation error). The required polynomial degree scales as
\begin{align}
d=\mathcal{O}\big(\bar{\alpha}' + \log(1/\varepsilon)\big), \label{eq:poly_qsp}
\end{align}
with all dependence on the coefficients $\omega_s$ absorbed into $\bar{\alpha}'$.

\paragraph{Circuit depth and mask dependence.}
The QSP circuit uses one signal qubit and $d$ controlled applications of $U_{\hat\sigma}$. Since $\mathrm{depth}(U_{\hat\sigma})=\mathcal O\big(\mathrm{depth} (U_{\texttt{prep}})+\ell_\sigma\,\mathrm{depth}(\text{c-}U_{\hat{\mathbb A}_s})\big)$, the overall depth scales as
\begin{align}
\mathrm{depth}\big(e^{\hat\sigma}\big)
=
\mathcal O\big(d\,\mathrm{depth}(U_{\hat\sigma})\big).
\end{align}
{\color{black}
A classical mask is incorporated by reparameterizing only $U_{\texttt{prep}}(\mathcal M^{(m)})$ while keeping the global normalization $\bar\alpha'$ fixed via the null branch. Consequently the \texttt{SELECT} table, ancilla wiring, and the QSP phase list (constructed for $\bar\alpha'$) are reused verbatim for every mask.}

Details of the QSP polynomial construction and error analysis are provided in Appendix~\ref{app:qsp_details}.

%%%%%%%%%%%%%%%%%%%%%%%%%%%%%%%%%%%%
\subsection{Similarity--sandwiched effective Hamiltonian}
\label{subsubsec:sim_sandwich}

With individual rank-one operators and their exponentials now in hand, we turn to encoding similarity-transformed Hamiltonians with respect to an anti-Hermitian generator and to projecting the result onto a reduced model space. In practice, after constructing the full exponential, we often wish to restrict the similarity transformation to only a subset of excitation operators—for example, to study active spaces of increasing size or to perform incremental downfolding. We capture such choices using a \emph{mask}: a classical bit string that specifies which rank-one terms are retained. The following paragraphs show how this mask is injected and how it propagates through the block-encoding pipeline.

\paragraph{Classical mask.}
Throughout the remainder of this section we assume that, at run time, a \emph{mask}
\begin{align}
\mathcal M^{(m)} \subset \{1,\dots,\ell_\sigma\}
\end{align}
is supplied, where the superscript $(m)$ labels the $m$-th choice of mask. The mask is a purely classical object that selects which rank-one generators $\{\hat{\mathbb A}_s\}$ (already block-encoded in Sec.~\ref{subsubsec:pqr1_exp_block}) are retained. Because the mask is classical, it incurs no coherence or ancilla overhead. Diagnostics for when the fixed rank-one algebraic structure remains stable, together with MP2-guided heuristics for constructing such masks in practice, are provided in Appendix~\ref{app:fixed_rank1}.

The truncated anti-Hermitian generator is therefore
\begin{align}
\hat\sigma^{(m)}
=
\sum_{s\in\mathcal M^{(m)}} -\,i\omega_s \hat{\mathbb A}_s,
~~
\hat{\mathbb A}_s = i(\hat L_s-\hat L_s^\dagger),
\end{align}
and the corresponding block encoding $U_{\hat\sigma}^{(m)}$ of $\exp(\hat\sigma^{(m)})$ is obtained by modifying only the state-preparation amplitudes in the selector register,
\begin{align}
U_{\texttt{prep}}(\mathcal M^{(m)})\ket{0^{\mathfrak{a}'}} = \ket{\chi^{(m)}}.
\end{align}
All other components of the circuit---including the \texttt{SELECT} table, ancilla registers, and QSP phase list---remain unchanged.

\paragraph{One-shot similarity sandwich.}
Given the binary-multiplexed Hamiltonian block encoding $W$ from Theorem~\ref{thm:sum_block}, we form the similarity-sandwiched unitary
\begin{align}
W_{\mathrm{eff}}^{(m)}
=
U_{\hat\sigma}^{(m)\dagger}\, W\, U_{\hat\sigma}^{(m)},
\label{eq:sim_trans_circuit}
\end{align}
which is unitary on the joint ancilla--system space and inherits the same normalization factor $\alpha$ as $W$. %Projecting all ancillas onto $\ket{0}$ and restricting the system to the mask-dependent model space yields
Projecting all ancillas onto $|0_{\mathrm{anc}}\rangle$ yields the normalized system-space block 
$(\langle 0_{\mathrm{anc}}|\otimes I)\,W_{\mathrm{eff}}^{(m)}\,(|0_{\mathrm{anc}}\rangle\otimes I)$.
We then define a user-specified \emph{model-space projector} $P^{(m)}$ whose range is a chosen subspace 
$\mathcal H_{\mathrm{mod}}^{(m)}:=\mathrm{Ran}(P^{(m)})\subseteq\mathcal H_{N_e}$ (e.g., an active-space determinant set or a truncated excitation manifold) consistent with the mask $\mathcal M^{(m)}$. In \texttt{COMPOSER}, $P^{(m)}$ is not assumed to be implemented as a coherent projector inside the oracle; rather, it specifies the subspace on which we represent the effective Hamiltonian (e.g., via matrix elements in a subspace-diagonalization routine). Accordingly, the encoded effective Hamiltonian guarantee is a bound on the \emph{restricted block}:
\begin{widetext}
\begin{align}
\Bigl\|P^{(m)}\Bigl[(\langle 0_{\mathrm{anc}}|\otimes I)\,W_{\mathrm{eff}}^{(m)}\,(|0_{\mathrm{anc}}\rangle\otimes I)
-\hat H_{\mathrm{eff}}^{(m)}/\alpha\Bigr]P^{(m)}\Bigr\|\le \epsilon'''.
\end{align}    
\end{widetext}
where {\color{black}
$P^{(m)}$ is defined as a user-specified model-space projector (e.g., an active-space determinant set or a truncated excitation manifold) that is consistent with the chosen mask $\mathcal M^{(m)}$.} Equivalently, for all $|\psi\rangle,|\phi\rangle\in\mathcal H_{\mathrm{mod}}^{(m)}$,
the matrix element error satisfies
\begin{align}
\bigl|\langle\phi|(\langle 0_{\mathrm{anc}}|W_{\mathrm{eff}}^{(m)}|0_{\mathrm{anc}}\rangle-\hat H_{\mathrm{eff}}^{(m)}/\alpha)|\psi\rangle\bigr|
\le \epsilon'''~. 
\end{align}
Here, the effective Hamiltonian is
\begin{align}
\hat H_{\mathrm{eff}}^{(m)}
=
P^{(m)} e^{-\hat\sigma^{(m)}} \hat H e^{\hat\sigma^{(m)}} P^{(m)}.
\label{eq:block_encode_H_eff}
\end{align}
The total block-encoding error obeys
\begin{align}
\epsilon''' = 2\epsilon'' + \epsilon,
\end{align}
where $\epsilon''$ is the error from the QSP approximation [Eq.~\eqref{eq:block_encode_exp_sigma}] and $\epsilon$ is the Hamiltonian block-encoding error from Theorem~\ref{thm:sum_block}. The factor of two arises from the two appearances of $U_{\hat\sigma}^{(m)}$ in the similarity sandwich. The error bound accumulates additively (up to constant factors independent of system size) with the number of logical composition layers (LCU, QSP, and the similarity sandwich), while remaining independent of the internal gate depth of each block encoding.

{\color{black}
\paragraph{Relation to Schrieffer--Wolff (SW) effective Hamiltonians.}
If $P$ projects onto a target model space and $Q=I-P$, an SW transformation chooses an anti-Hermitian $\hat\sigma$ such that the transformed Hamiltonian is (approximately) block diagonal, $Q\,e^{-\hat\sigma}He^{\hat\sigma}\,P \approx 0$, yielding an effective model-space Hamiltonian $H_{\mathrm{eff}}=P e^{-\hat\sigma}He^{\hat\sigma}P$. In \texttt{COMPOSER}, we do not prescribe how $\hat\sigma$ is obtained (perturbative SW, variational fitting, or flow-based updates); instead we provide a compile-once oracle that evaluates the masked sandwich $P^{(m)}e^{-\hat\sigma^{(m)}}H e^{\hat\sigma^{(m)}}P^{(m)}$ while keeping the two-qubit topology fixed. The classical mask $\mathcal{M}^{(m)}$ may be interpreted as selecting which couplings are actively eliminated (or retained) in a truncated SW generator, enabling systematic model-space refinement without circuit recompilation.}\\

\noindent {\bf Remark.} If one wishes to enforce model-space restriction coherently inside the oracle (rather than treating $P^{(m)}$ as a specification of which matrix elements are queried), one may introduce an additional membership oracle for $\mathcal H_{\mathrm{mod}}^{(m)}$ and use it to (i) flag leakage out of the model space or (ii) postselect/amplify within $\mathcal H_{\mathrm{mod}}^{(m)}$. Such enforcement is application dependent and is outside the scope of the compile-once block-encoding construction presented here.

\paragraph{Mask updates require \emph{only} re-dialing angles.}
Switching from $\mathcal M^{(m)}$ to a new mask $\mathcal M^{(m')}$ leaves the selector register, the \texttt{SELECT} table $W_{\texttt{sel}}$, the QSP phase list $\{\phi_k\}$, and all ancilla wiring untouched. Only the rotation angles in the deterministic state-preparation ladder that prepares $\ket{\chi^{(m')}}$ are updated. Consequently, the circuit topology is compiled once and reused verbatim across all mask choices. The overall masked similarity-sandwich workflow is sketched in Figure~\ref{fig:mask_workflow_qtikz}.

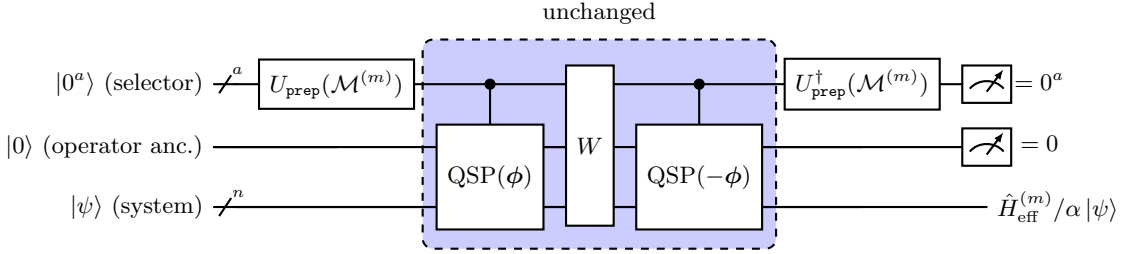
\begin{figure*}
\centering
\begin{quantikz}[row sep=0.2cm, column sep=0.3cm, font=\small]
%------------- selector / LCU layer --------------------
\lstick{$\ket{0^{a}}$ (selector)} 
& \qwbundle{a} 
& \gate{U_{\texttt{prep}}(\mathcal M^{(m)})} 
& \ctrl{2}\gategroup[3,steps=3,style={dashed,rounded corners,fill=blue!20, inner xsep=2pt},background,label style={anchor=north,yshift=0.4cm}]{unchanged}  
& \gate[wires=3]{W} 
& \ctrl{2} 
& \gate{U_{\texttt{prep}}^{\dagger}(\mathcal M^{(m)})} 
& \meter[label style={label position=below,anchor=north,xshift=0.7cm}]{=0^a} \\
%------------- QSP ladder ------------------------------
\lstick{$\ket{0}$ (operator anc.)} 
& \qw 
& \qw 
& \gate[wires=2]{\text{QSP}(\boldsymbol{\phi})} 
& \qw 
& \gate[wires=2]{\text{QSP}(-\boldsymbol{\phi})} 
& \qw 
& \meter[label style={label position=below,anchor=north,xshift=0.7cm}]{=0} \\
%------------- system wire ------------------------------
\lstick{$\ket{\psi}$ (system)} 
& \qwbundle{n} 
& \qw 
& \qw 
& \qw 
& \qw 
& \qw 
& \rstick{$\hat H_{\text{eff}}^{(m)}/\alpha \ket{\psi}$}
\end{quantikz}
\caption{Compile-once circuit for the similarity-sandwiched effective Hamiltonian. Changing the classical mask \(\mathcal M^{(m)}\) only redials the rotation angles in \(U_{\texttt{prep}}(\mathcal M^{(m)})\); every other gate, ancilla, and QSP phase is reused verbatim for every \emph{mask}.}
\label{fig:mask_workflow_qtikz}
\end{figure*}
%%%%%%%%%%%%%%%%%%%%%%%%%%%%%%%%%%%%

\subsection{Resource analysis and practical advantages}
\label{subsec:rank1_resources}

The block encoding of the effective Hamiltonian constructed above involves two independent linear combinations:
\begin{itemize}
    \item a \emph{masked generator}
    \(
      \hat\sigma^{(m)}=\sum_{s\in\mathcal M^{(m)}}-i\omega_s\hat{\mathbb A}_s
    \)
    containing
    \(\ell_\sigma\)
    rank-one terms; and
    \item the physical Hamiltonian
    \(
      \hat H=\sum_{s=1}^{\ell_H}\Omega_s\hat L_s
    \),
    factorized into \(\ell_H\) rank-one ladders.
\end{itemize}
We keep these symbols distinct in the resource estimates below.

\medskip\noindent\textbf{Gate counts.}
The selector state-preparation ladder \(U_{\texttt{prep}}(\mathcal M^{(m)})\) is implemented using controlled single-qubit rotations and CNOT gates. Its depth scales as \(\mathcal O(\ell_\sigma)\), arising solely from the number of retained labels.

To keep the depth expressions uniform across ladder types, we define the selector-controlled two-qubit depth of a single adaptor branch by its ladder class:
\begin{enumerate}
    \item $D^{(I)}(n) := \mathrm{depth}(c\text{-}W_s)$ for bilinear ladders (Lemma~\ref{lem:rank1_block});
    \item $D^{(III)}(n,R_\mu) := \mathrm{depth}(c\text{-}W_\mu)$ for Cholesky-channel ladders (Lemma~\ref{lem:quad_rank1}).
\end{enumerate}
For pair-excitation ladders, we denote the corresponding selector-controlled depth by
\[
D^{(II)}(n;\mathcal G),
\]
which depends on the hardware connectivity graph $\mathcal G$ through the routing/scheduling of the four-qubit pair--Givens.
Each rank-one adaptor branch acts on the same $n$-qubit system register and has selector-controlled depth $D^{(I)}(n)=\mathcal O(n)$ for bilinear ladders (Lemma~\ref{lem:rank1_block}), $D^{(III)}(n,R_\mu)=\mathcal O(n+R_\mu)$ for diagonalized Cholesky-channel ladders (Lemma~\ref{lem:quad_rank1}), and $D^{(II)}(n;\mathcal G)$ for pair-excitation ladders. In particular, Appendix~\ref{app:A.3} (Table~\ref{tab:B_network}) shows that $D^{(II)}(n;\mathcal G)$ scales as $\mathcal O(n^2)$ on linear/heavy-hex connectivity and as $\mathcal O(n)$ on 2D grid and all-to-all connectivity under standard scheduling assumptions.
Because the selector-controlled adaptors are mutually exclusive, the total depth of the generator \texttt{SELECT} stage scales as
\begin{align}
\mathrm{depth}\!\left(W_{\mathrm{sel}}(\sigma)\right)
&= \mathcal O\!\left(\sum_{s=1}^{\ell_\sigma} \mathrm{depth}(c\text{-}W_{\hat A_s})\right) \notag \\
&= \mathcal O\!\left(\ell_\sigma\,D_{\sigma}^{\max}\right),
\end{align}
where $D_{\sigma}^{\max}:=\max_{s\in\{1,\dots,\ell_\sigma\}} \mathrm{depth}(c\text{-}W_{\hat A_s})$. For generators dominated by pair-excitation ladders, $D_{\sigma}^{\max}\approx D^{(II)}(n;\mathcal G)$, whereas if only
bilinear ladders are used then $D_{\sigma}^{\max}=D^{(I)}(n)=\mathcal O(n)$.

For the Hamiltonian block encoding, we use the same multiplexing construction with \(\ell_{\hat H}\propto n\) rank-one terms arising from the nested factorization. This yields
\begin{align}
\mathrm{depth}\!\left(W_{\mathrm{sel}}(\hat H)\right)
= \mathcal O\!\left(\sum_{s=1}^{\ell_H} \mathrm{depth}(c\text{-}W_s)\right), \label{eq:Wsel_H}
\end{align}
independent of the generator mask. In the nested-factorized Hamiltonian of Eq.~\eqref{eq2:quadratic_rank_1}, Eq.~\eqref{eq:Wsel_H} decomposes into a bilinear part (one-electron modes) and a quadratic-channel part (Cholesky channels), giving a representative scaling 
\begin{align}
    \mathcal O\!\left(R_1D^{(I)}(n)+\sum_{\mu=1}^{K} D^{(III)}(n,R_\mu)\right).
\end{align}

Applying quantum signal processing with polynomial degree $d$,  the full similarity-sandwiched oracle therefore has depth
\begin{align}
\mathrm{depth}\!\left(W_{\mathrm{eff}}^{(m)}\right)
&=\mathcal O\!\left(d\,\mathrm{depth}(U_{\hat\sigma}^{(m)})+\mathrm{depth}(W)\right) \notag \\
&=\mathcal O\!\left(d\,\ell_\sigma\,D_{\sigma}^{\max} + \mathrm{depth}(W_{\mathrm{sel}}(\hat H))\right),
\end{align}
Here the $d\,\ell_\sigma D_{\sigma}^{\max}$ contribution arises exclusively from the QSP-controlled applications of the generator block encoding $U_{\hat\sigma}$, while the $n^2$ term originates from the Hamiltonian \texttt{SELECT} stage and is not multiplied by the QSP degree. In particular, the Hamiltonian block encoding is invoked only once inside the similarity sandwich.

\medskip\noindent\textbf{Ancilla budget.}
The generator and Hamiltonian linear combinations require selector registers of widths $\mathfrak{a}_\sigma$ and $\mathfrak{a}_H$, respectively. Because these ladders are executed sequentially and the selector register is reset between them, a single physical register of width $\mathfrak{a}=\max(\mathfrak{a}_\sigma,\mathfrak{a}_H)$ suffices. Including adaptor-internal ancillas, the total ancilla width is $\mathfrak{a} + t$, where $\mathfrak{a}=\max(\mathfrak{a}_\sigma,\mathfrak{a}_H)$ is the selector width and $t$ is the maximum number of oracle ancillas required by any selected adaptor (e.g., $t=1$ for bilinear dyads, and {\color{black}$t=\mathfrak{a}_I+2$} for the diagonalized Cholesky-channel adaptor in Lemma~\ref{lem:quad_rank1}).

The architectural implications of this scaling—namely, freezing the logical two-qubit topology and streaming only parametric updates—are discussed in detail in Sec.~\ref{sec:composer}.

\medskip\noindent\textbf{Quantitative payoff: what is avoided under instance updates.}
Let $\ell_H^{\mathrm{pool}}$ and $\ell_\sigma^{\mathrm{pool}}$ denote the sizes of the \emph{compiled} rank-one term pools for $\hat H$ and $\hat\sigma$, chosen large enough to cover a target family of instances (e.g., a geometry scan or active-space growth), and let $\ell_\sigma=|\mathcal M^{(m)}|$ denote the \emph{active} masked generator size at a particular update.
A conventional instance-specific pipeline typically regenerates the term list (or its data-access structure), rebuilds the corresponding multiplexing and routing layers, and re-runs hardware mapping/routing whenever coefficients, truncations, or active spaces change.
In \texttt{COMPOSER}, the two-qubit fabric implementing the selector-controlled adaptors and routing is compiled once for $(\ell_H^{\mathrm{pool}},\ell_\sigma^{\mathrm{pool}})$; a geometry or mask update only \emph{dials} new single-qubit angles (\texttt{PREP} amplitudes and local ladder phases) and streams new classical coefficients/masks.
A concrete proxy is the number of two-qubit layers whose placement/routing is compiled once: from Table~\ref{tab:hardware}, the logical two-qubit depth of one masked similarity-sandwiched oracle scales as $\mathcal O(d\,n\,\ell_\sigma^{\mathrm{pool}}+n^2)$ and is unchanged across updates, whereas the per-instance compilation cost of an instance-specific build scales with this entire two-qubit fabric \cite{Sivarajah2020tket,Cowtan2020QubitRouting}.
Table~\ref{tab:compile_once_payoff} summarizes which objects must be regenerated in a conventional pipeline versus \texttt{COMPOSER}.

\begin{table*}
\centering
\caption{Logical-level resource scaling for one masked similarity-sandwiched oracle. Depth counts two-qubit layers; single-qubit rotations are treated as depth-1.}
\label{tab:hardware}
\begin{tabular}{lccc}
\\[-2ex] \hline \\[-2ex]
Circuit block & System qubits & Ancillas & Depth \\
\hline \\[-2ex]
%Adaptor for one $\hat L_s$ & $n$ & $1$ (signal) & $\mathcal O(n)$ \\
Adaptor (bilinear dyad) & $n$ & $1$ (signal) & $\mathcal O(n)$ \\
Adaptor (pair-excitation ladder) & $n$ & $1$ (signal) & $D^{(II)}(n;\mathcal G)$ (see Table~\ref{tab:B_network}) \\
Adaptor (Cholesky channel) & $n$ & $\mathfrak{a}_I+2$ & $\mathcal O(n+R_\mu)$ \\
Hamiltonian \texttt{SELECT}
& $n$ & $\mathfrak{a}_H=\lceil\log_2 \ell_H\rceil$ & $\mathcal O(n^2)$ \\
Generator \texttt{SELECT}
& $n$ & $\mathfrak{a}_\sigma=\lceil\log_2 \ell_\sigma\rceil$ & $\mathcal O(\ell_\sigma D^{max}_\sigma)$ \\
\texttt{PREP} amplitude ladder
& -- & $\max(\mathfrak{a}_\sigma,\mathfrak{a}_H)$ & $\mathcal O(\ell_\sigma)$ \\
Two QSP ladders
& $n$ & shared signal qubit & $\mathcal O(d \ell_\sigma D^{max}_\sigma)$ \\[1ex]
\hline \\[-2ex]
\textbf{Total}
& $n$ & $\mathfrak{a}+t$ & $\mathcal O(d \ell_\sigma D^{max}_\sigma + {\rm depth}(W_{\rm sel}(\hat H)))$ \\
\hline
\end{tabular}
\end{table*}

\begin{table*}[t]
\centering
\caption{What must be regenerated under instance updates (geometry/mask/truncation) in a conventional instance-specific pipeline versus \texttt{COMPOSER}. ``Regenerate'' indicates structural recompilation/remapping/rerouting of two-qubit layers; ``dial'' indicates updating classical coefficients and single-qubit rotation angles while reusing the same two-qubit fabric.}
\label{tab:compile_once_payoff}
%\resizebox{0.5\textwidth}{!}{
\begin{tabular}{lcc}
\hline
Artifact under updates & Conventional build & \texttt{COMPOSER} \\
\hline
Term list / truncation pattern & regenerate & fixed pool + classical mask \\
Data-loading for coefficients & regenerate & dial for the same topology \\
\texttt{SELECT} multiplexer and two-qubit routing & regenerate & compiled once \\
\hline
\end{tabular}
%}
\end{table*}

%%%%%%%%%%%%%%%%%%%%%%%%%%%%%%%%%%%%%

\begin{figure*}
\centering
\includegraphics[width=\linewidth]{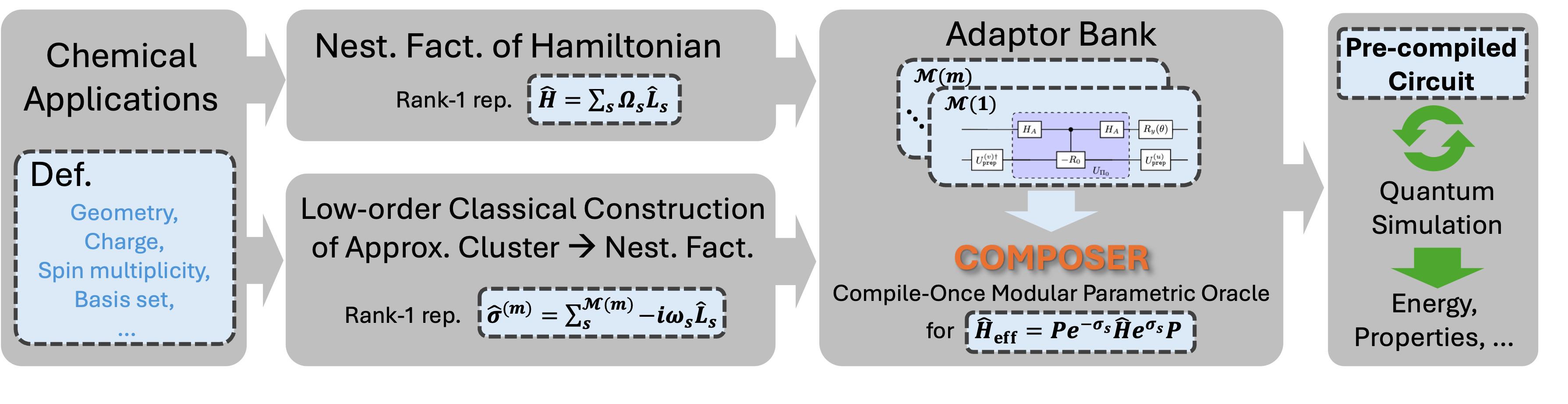}
\caption{Schematic workflow for quantum simulation using \texttt{COMPOSER}. The two-qubit circuit topology is fixed in a one-time synthesis pass; subsequent updates to molecular geometry, active space, or truncation mask re-dial only single-qubit rotations.}
\label{fig:composer_arch}
\end{figure*}

\section{Architecture, and Compile--Once Philosophy of \texttt{COMPOSER}}
\label{sec:composer}

Having established the operator factorizations, adaptor constructions, and resource scalings, we now elevate the discussion to the architectural level. The nested factorizations of $\hat H$ and $\hat\sigma$ compress both the Hamiltonian and similarity generator into linearly scaling collections of rank-one operators. 
Each rank-one ladder $\hat L_s$ is converted into a deterministic block encoding whose selector-controlled depth depends on its ladder class: $\mathcal O(n)$ for bilinear ladders, $\mathcal O(n+R_\mu)$ for diagonalized Cholesky-channel ladders, and $D^{(II)}(n;\mathcal G)$ for pair-excitation ladders, where the latter is connectivity dependent (Appendix~\ref{app:A.3}).
These ladders are assembled within a fixed \texttt{PREP-SELECT-PREP}$^\dagger$ skeleton that defines a reusable two-qubit circuit fabric (Figure~\ref{fig:mask_workflow_qtikz}).

The central design principle of \texttt{COMPOSER} is to disentangle circuit \emph{topology} from numerical \emph{data}. Topology comprises the selector cascade, adaptor layout, and QSP scaffold; data comprises the coefficients and rotation angles that vary with geometry, active space, or mask choice. By freezing the selector-controlled two-qubit structure and allowing only parametric single-qubit updates, \texttt{COMPOSER} shifts adaptation cost from structural recompilation to classical parameter streaming. Importantly, QSP overhead is confined to exponentiation of masked generators; the Hamiltonian block encoding itself remains a single multiplexed layer whose depth is independent of the QSP degree.
Figure~\ref{fig:mask_workflow_qtikz} summarizes the resulting two-phase workflow in a way that separates the one-time compilation work from the per-instance parameter updates (see also Algorithm~\ref{alg:composer_compile_dial} in Appendix~\ref{app:composer_impl}).

\paragraph{Adaptor abstraction.}
Each rank-one operator is associated with a fixed-topology adaptor that implements its normalized block encoding. Adaptors are tagged by binary addresses and embedded in an invariant selector tree determined solely by the total number of rank-one terms. Changes in coefficients or truncation masks modify only \texttt{PREP} amplitudes and local phases; the underlying two-qubit connectivity graph remains unchanged.

\paragraph{Compile--once execution model.}
\texttt{COMPOSER} generates the logical two-qubit schedule of this architecture once \cite{Sivarajah2020tket,Cowtan2020QubitRouting,Zulehner2018MappingIBM}. Subsequent instance updates—geometry sweeps, active-space growth, or mask refinement—stream new single-qubit angles to the same fixed graph. If a subset of rank-one terms is disabled by a classical mask, the corresponding adaptors are bypassed without altering circuit structure. Here ``compile-once'' refers specifically to invariance of the logical two-qubit topology; hardware calibration and fault-tolerant synthesis layers remain orthogonal \cite{Kliuchnikov2013CliffordTApprox,Selinger2015CliffordT,Amy2013MeetInMiddleSynthesis}.

\paragraph{Subspace diagonalization workloads (GCIM/QSE/NOQE) as a natural dial-many target.}
{\color{black}
The compile-once premise is particularly well aligned with subspace diagonalization workflows in which one repeatedly evaluates \emph{families} of closely related state-preparation circuits to build projected matrices $H$ and $S$.
In QSE/NOQE/GCIM-style solvers, basis states are often generated by exponentiating a small set of anti-Hermitian operators drawn from a common pool, and the working subspace is expanded by activating additional generators or adjusting their scalar weights~\cite{McClean2017QSE,Huggins2020NonOrthogonalVQE,Baek2023NOQE,Zheng2023QuantumGCM,Zheng2024UnleashedGCIM}.
In \texttt{COMPOSER} such updates correspond precisely to re-dialing the \texttt{PREP} angles and classical masks/selectors while leaving the two-qubit fabric (selector cascade, adaptor layout, and QSP scaffold) unchanged.
This enables a unified compiled circuit template for both expectation-value measurements and overlap measurements between different masked basis states, even as the subspace is adaptively enlarged.
}

\paragraph{Fault-tolerant angle synthesis (what compile-once does and does not buy).}
The compile-once promise of \texttt{COMPOSER} is about \emph{logical topology} (the two-qubit connectivity pattern and routing plan). In a fault-tolerant setting, this logical circuit must still be lowered to a discrete gate set (e.g., Clifford+$T$), and the cost of synthesizing the updated single-qubit rotations can dominate when high precision is required. Thus, compile-once should be read as \emph{eliminating topology recompilation} (remapping/rerouting/re-optimizing the two-qubit fabric) across instance updates, while leaving rotation-synthesis overhead as an orthogonal, accuracy-controlled layer \cite{Kliuchnikov2013CliffordTApprox,Selinger2015CliffordT,Amy2013MeetInMiddleSynthesis}. In practice, the rotation-synthesis tolerance can be incorporated into the same error budget that already allocates block-encoding and QSP approximation errors (Appendix~\ref{app:composer_impl}).

\paragraph{One canonical closed-loop workflow for updating $\hat\sigma$ and masks.}
A typical usage pattern is: (i) choose an orbital pool/basis and compile the circuit skeleton once (Algorithm~\ref{alg:composer_compile_dial}); (ii) for each geometry or active-space choice, compute integral-derived low-rank factors for $\hat H$ and obtain an inexpensive proxy for the similarity generator (e.g., MP2- or low-order SW-seeded amplitudes), then screen to define a mask $\mathcal M^{(m)}$; (iii) run the fixed oracle to evaluate energies/observables in the current model space; and (iv) refine by expanding the mask/active space or updating the generator parameters, while reusing the same compiled two-qubit fabric. Across this loop, the selector tree, adaptor wiring, and QSP scaffold remain fixed; only the classical coefficients/masks and the corresponding single-qubit angles are updated.

\paragraph{Error budget.}
The total tolerance $\epsilon_{\mathrm{tot}}$ decomposes into factorization, block-encoding, multiplexing, and QSP contributions. Separately, any error due to the choice of encoded operators (e.g., a CCSD-seeded truncated generator or a particular mask) is a modeling/initialization consideration and can be improved by refining the encoded data without changing the compiled two-qubit topology. This separation exposes independent control knobs for depth--accuracy tradeoffs without modifying the frozen topology (Appendix~\ref{app:composer_impl}).

\medskip
\texttt{COMPOSER} therefore reorganizes standard block-encoding primitives into a topology-invariant execution paradigm. Rather than altering asymptotic query complexity, it enforces structural reuse across families of closely related operators, enabling predictable resource scaling in adaptive quantum workflows.

%%%%%%%%%%%%%%%%%%%%%%%%%%%%%%%%%%%%%

\section{Relation to Local Jastrow Ansatz, Tensor Networks, and Canonical Transformations}
\label{sec:JastrowTensor}

The rank-one quadratic occupation operators arising from the nested factorization in \texttt{COMPOSER} bear a close relationship to both \emph{local Jastrow} correlation factors and \emph{tensor-network} operator representations. These connections provide additional physical and structural intuition for the resulting block-encoding architecture, but are not required for its construction.

\paragraph{Connection to the local Jastrow ansatz.}
A (generally \emph{non-unitary}) Jastrow correlation factor in lattice/second-quantized form is often written as \cite{Jastrow1955}
\begin{equation}
  e^{\sum_{p<q} u_{pq} \, n_p n_q},
\end{equation}
where $n_p = a_p^\dagger a_p$ is the occupation operator of orbital $p$.
{\color{black}For real $u_{pq}$, this operator is diagonal in the occupation basis and preserves particle number, but it \emph{reweights} configuration amplitudes rather than applying phases.}
{\color{black}A \emph{unitary} (phase-imprinting) variant frequently used in quantum-circuit contexts replaces the exponent by $-i\sum_{p<q}\theta_{pq}n_p n_q$, yielding a diagonal unitary that applies configuration-dependent phases \cite{Motta2023LUCJ,MatsuzawaKurashige2020JastrowDecomp,Stenger2023JastrowGutzwiller}.}

In \texttt{COMPOSER}, the projected quadratic rank-one operators [Eq.~\eqref{eq2:quadratic_rank_1}] take a closely related form:
{\color{black}they are diagonal \emph{in a channel-dependent rotated orbital basis} (e.g., the eigenbasis of each Cholesky/DF channel),}
and preserve particle number, but their coefficients are not free variational parameters—they emerge directly from the electronic-structure data and chosen factorization thresholds.
{\color{black}In the original computational basis, these operators become diagonal only after conjugation by the corresponding single-particle rotation (implemented by the same deterministic ladder/Givens primitives used elsewhere in \texttt{COMPOSER}).}
{\color{black}If one chooses to exponentiate such diagonal quadratic forms as part of a similarity generator (e.g., unitary cluster-Jastrow layers), they generate Jastrow-like correlation patterns with \emph{data-driven} coefficients \cite{Motta2023LUCJ,MatsuzawaKurashige2020JastrowDecomp}.}

\paragraph{Connection to tensor networks.}
The nested factorization of $\hat H$ and $\hat\sigma$ produces sums of outer products of low-dimensional coefficient vectors, which is closely related to the low-rank decompositions used to obtain compact matrix product operator (MPO) forms~\cite{Verstraete2008TNReview,Orus2014TNReview,Crosswhite2008MPO,White1992DMRG,Schollwock2011DMRG,ChanSharma2011DMRGQChem,Keller2015MPOChem}. From this perspective:
\begin{itemize}
    \item the selector register in \texttt{COMPOSER} plays the role of a virtual (bond) index that labels which low-rank operator branch is active;
    {\color{black}\item the overall \texttt{PREP-SELECT-PREP}$^\dagger$ construction is naturally interpreted as a \emph{star-shaped} (LCU/CP-like) operator tensor network rather than a strictly site-local MPO chain;}
    \item each rank-one adaptor corresponds to a structured operator tensor associated with one selector value; and
    \item the fixed \texttt{PREP-SELECT-PREP}$^\dagger$ skeleton defines the network connectivity, while single-qubit rotation angles supply the tensor entries.
\end{itemize}
This viewpoint highlights that \texttt{COMPOSER} can be regarded as a hardware-native, fixed-topology tensor-network \emph{execution} skeleton whose numerical content is updated efficiently via classical streaming of angles rather than circuit recompilation.

\paragraph{Connection to Schrieffer--Wolff and canonical transformations.}
{\color{black}Schrieffer--Wolff (SW)-type downfolding and related canonical-transformation formalisms can be viewed as generating effective interactions (after similarity transformation and projection) that are often dominated by density--density and other structured low-body terms in an appropriate basis \cite{SchriefferWolff1966,Bravyi2011SchriefferWolff,YanaiChan2007CanonicalTransformation}.}
{\color{black}The Jastrow-like (diagonal) and tensor-network-like (low-rank) perspectives above therefore provide an intuitive bridge: \texttt{COMPOSER} supplies a compile-once block-encoding fabric that can repeatedly \emph{update} such effective couplings across masks/model spaces without recompiling two-qubit routing, aligning naturally with iterative SW-style workflows.}

Conceptually, \texttt{COMPOSER} unifies (i) diagonal, number-conserving quadratic structure reminiscent of Jastrow correlators; (ii) structured low-rank decompositions characteristic of tensor networks; and (iii) a depth-optimized, compile-once quantum circuit topology. This synthesis enables structured effective interactions to be encoded within a tensor-network-inspired operator skeleton while retaining the flexibility and resource efficiency of the mask-aware execution model.

%%%%%%%%%%%%%%%%%%%%%%%%%%%%%%%%%%%%%
\section{Numerical Demonstration}
\label{sec:numerics}

\begin{figure}[t]
\centering
\includegraphics[width=\linewidth]{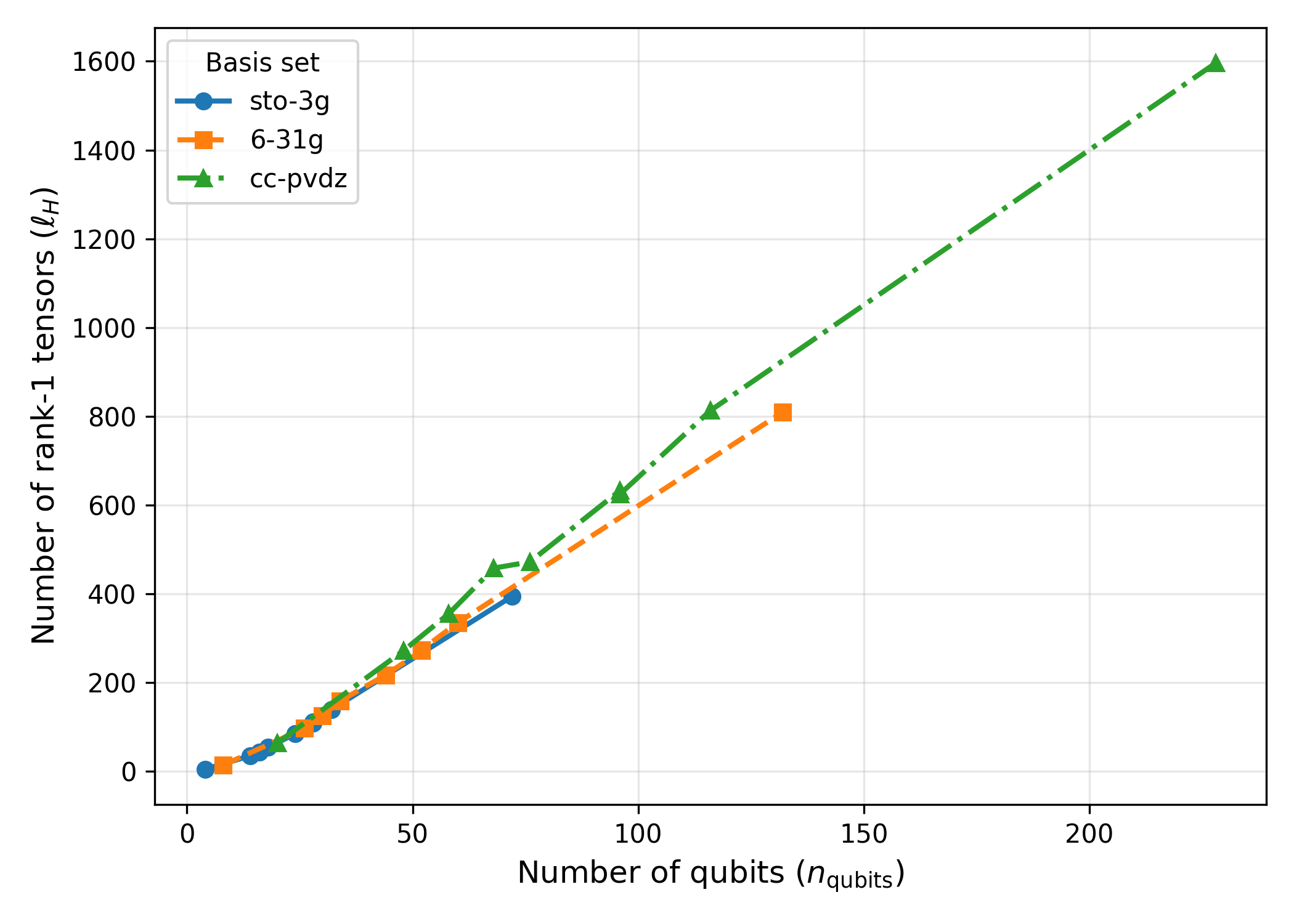}
\caption{Scaling of the total number of rank-1 tensors $\ell_H$ required for Hamiltonian block encoding as a function of the number of qubits $n_{\mathrm{qubits}}$. Results are shown for three commonly used Gaussian basis sets: STO-3G, 6-31G, and cc-pVDZ. Each data point corresponds to a molecular system included in the benchmark set, with $n_{\mathrm{qubits}}$ determined by the number of spin orbitals in the chosen basis. The approximately linear trend over the tested range indicates empirically quasi-linear growth of the rank-1 ladder pool with system size, while the vertical offsets reflect the increased two-electron integral complexity of larger basis sets. {\color{black}
All data use a fixed Cholesky threshold $\tau_{\mathrm{chol}}=10^{-8}$ (and channel eigenvalue cutoff $\tau_{\mathrm{eig}}=10^{-4}$), so that the observed quasi-linear growth reflects empirical scaling of the required Cholesky-channel count $K$ with system size for chemically relevant tolerances.}
}
\label{fig:rank1_scaling_nqubits}
\end{figure}

\begin{figure*}
\centering
\includegraphics[width=\linewidth]{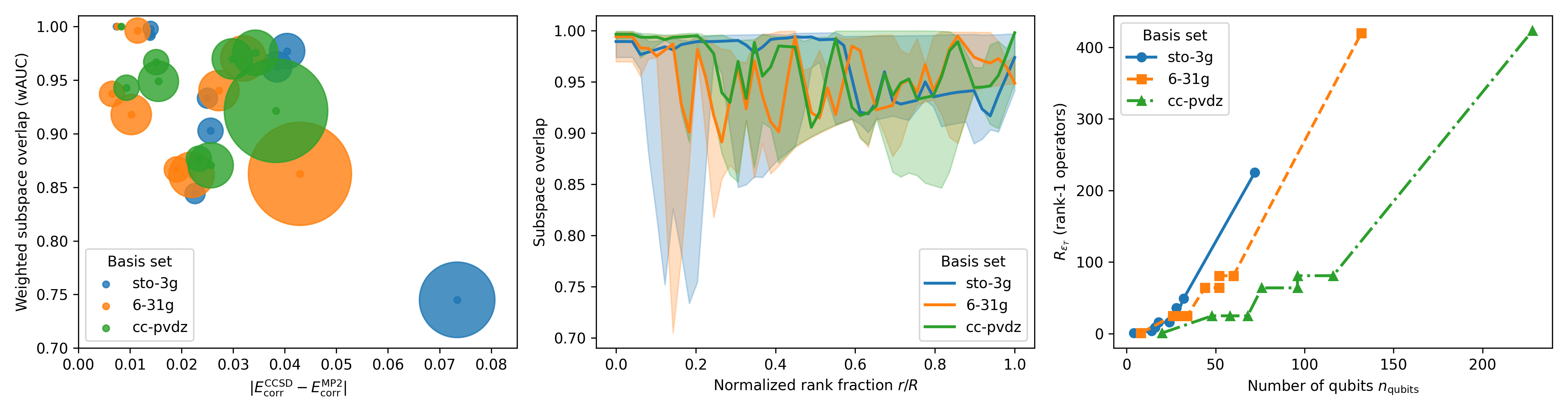}
\caption{Cross-molecular validation of MP2-derived rank-1 operator subspaces against CCSD across multiple basis sets. (a) Weighted average subspace overlap (wAUC) between MP2- and CCSD-derived rank-1 spaces for each molecule and basis set (STO-3G, 6-31G, and cc-pVDZ), plotted against the correlation-energy discrepancy $|E_{\mathrm{corr}}^{\mathrm{CCSD}} - E_{\mathrm{corr}}^{\mathrm{MP2}}|$. Marker size encodes the number of retained rank-1 operators $R_{\varepsilon}$, defined by a global relative singular-value screening criterion $s_k/s_1 \ge \varepsilon_s$ applied to the CCSD $T_2$ spectrum. (b) Median MP2--CCSD subspace overlap as a function of normalized rank fraction $r/R$, with shaded regions indicating the interquartile range across all molecules for each basis set. (c) Scaling of the number of retained rank-1 operators $R_{\varepsilon}$ with the number of qubits $n_{\mathrm{qubits}}$, illustrating the growth of the dominant $T_2$ operator manifold with system size and basis-set quality.}
\label{fig:rank1_overlap_multi_basis}
\end{figure*}

The numerical tests reported in this section use a small benchmark set of closed-shell molecules spanning increasing system size and bonding complexity: H$_2$, H$_2$O, NH$_3$, CH$_4$, H$_2$CO, C$_2$H$_4$, CH$_3$OH, C$_2$H$_6$, and C$_6$H$_6$. These systems include diatomic, triatomic, and polyatomic species with single, double, and aromatic bonding motifs, and involve C, H, O, and N atoms. All geometries were fixed at standard equilibrium structures. For each molecule, calculations were performed in the STO-3G, 6-31G, and cc-pVDZ basis sets. 

{\color{black}
Unless stated otherwise, all calculations use restricted Hartree--Fock canonical orbitals and a fixed electron number with no symmetry breaking. Two-electron integrals are factorized via pivoted Cholesky with threshold $\tau_{\mathrm{chol}}=10^{-8}$, which determines the number of Cholesky channels $K$ and hence the rank-one Hamiltonian pool size $\ell_H$ [Eq.~\eqref{eq:Ham_rank1}]. For each channel, eigenmodes with $|\lambda^{(\mu)}_\xi|/\max_\xi|\lambda^{(\mu)}_\xi|<\tau_{\mathrm{eig}}=10^{-4}$ are discarded, defining $R_\mu$. MP2 amplitudes follow M\o ller--Plesset second-order perturbation theory, and CCSD amplitudes are obtained from converged coupled-cluster singles and doubles equations with energy residual tolerance $10^{-6}$ a.u.}

We present numerical evidence that the \emph{structural conditions required for topology-invariant compilation} are satisfied in representative molecular workflows. Specifically, we validate (i) that the nested Hamiltonian factorization yields a rank-one ladder pool whose size grows mildly with system size (so the \texttt{SELECT} fabric can be fixed), and (ii) that low-rank excitation subspaces are stable under inexpensive classical proxies (so \emph{mask updates} can be realized by re-dialing \texttt{PREP} amplitudes without restructuring the circuit). Accordingly, the goal of this section is \emph{not} to benchmark chemical energies or fault-tolerant Toffoli counts, but to support the architectural premise of Figure~\ref{fig:motivation}: across families of closely related instances (basis choice, molecule size, and masking), the logical two-qubit topology can remain unchanged while instance dependence enters only through streamed single-qubit parameters.

\paragraph{Scaling of rank-one Hamiltonian factorizations.}
Figure~\ref{fig:rank1_scaling_nqubits} shows the total number of rank-one tensors $\ell_H$ required to represent the electronic Hamiltonian as a function of the number of qubits $n_{\mathrm{qubits}}$, for three commonly used Gaussian basis sets (STO-3G, 6-31G, and cc-pVDZ). {\color{black}
Here $\ell_H$ denotes the number of \emph{rank-one adaptor branches} in the Hamiltonian LCU, i.e., $\ell_H=R_1+K$ in Eq.~\eqref{eq:Ham_rank1}, where $R_1$ is the number of retained one-electron modes (typically $R_1\approx n_{\mathrm{qubits}}$ in the absence of additional truncation) and $K$ is the number of Cholesky channels at threshold $\tau_{\mathrm{chol}}$. Thus, the nontrivial scaling signal is primarily the quasi-linear growth of $K$ for fixed $\tau_{\mathrm{chol}}$.}
Each data point corresponds to a molecular system included in the benchmark set, with $n_{\mathrm{qubits}}$ determined by the number of spin orbitals in the chosen basis.

Across all basis sets, $\ell_H$ exhibits empirically quasi-linear growth over the tested molecular range, with basis-set-dependent slopes. Since the selector width is $\mathfrak{a}_H=\lceil\log_2 \ell_H\rceil$, this implies only logarithmic growth in the selector register even as the ladder pool expands. As a result, a fixed \texttt{PREP-SELECT-PREP}$^\dagger$ multiplexing topology can be chosen once and reused across increasing molecular complexity, satisfying a central structural requirement of the compile-once execution model.

{\color{black}
\paragraph{Instance-update invariance under geometry scans (compile-once stress test).}
We performed a short geometry scan for H$_2$O by scaling the equilibrium H--O bond length $R_e$ by factors $0.5,1.0,1.5,$ and $2.0$ and recomputing the nested Hamiltonian factorization at each geometry. Table~\ref{tab:h2o_scan_lH} reports the Cholesky-channel count $K$ and the resulting Hamiltonian ladder count $\ell_H=R_1+K$ across three basis sets. Over the full scan window, $\ell_H$ varies mildly (STO-3G: constant at 35; 6-31G: 94--97; cc-pVDZ: 239--272), so the selector width can be fixed once as $\mathfrak{a}_H=\lceil\log_2 \ell_H^{\max}\rceil$ without changing any two-qubit multiplexing/routing structure. Geometry dependence is absorbed entirely into updated single-qubit angles in \texttt{PREP} and local ladder phases, consistent with the compile-once premise. In particular, for each basis we can fix $\mathfrak a_H=\lceil\log_2\ell_H\rceil$ at the maximum value over the scan window and treat all geometry dependence as streamed updates to the \texttt{PREP} angles.}

\begin{table*}
\centering
\caption{H$_2$O geometry-scan stress test for topology-invariant compilation. We vary the H--O bond length from $0.5R_e$ to $2.0R_e$ (with $R_e$ the equilibrium bond length) and report the Cholesky-channel count $K$ and total Hamiltonian rank-one ladder count $\ell_H=R_1+K$ (shown as $K/\ell_H$). The compiled selector width can be fixed for each basis as $\mathfrak{a}_H=\lceil\log_2 \ell_H^{\max}\rceil$ over the scan window.}
\label{tab:h2o_scan_lH}
\begin{tabular}{lcccccc}
\hline
Basis & $0.5R_e$ & $1.0R_e$ & $1.5R_e$ & $2.0R_e$ & $\ell_H^{\max}$ & $\mathfrak{a}_H=\lceil\log_2 \ell_H^{\max}\rceil$ \\
\hline
sto-3g  & 28/35  & 28/35  & 28/35  & 28/35  & 35  & 6 \\
6-31g   & 84/97  & 84/97  & 84/97  & 81/94  & 97  & 7 \\
cc-pVDZ & 232/256 & 248/272 & 230/254 & 215/239 & 272 & 9 \\
\hline
\end{tabular}
\end{table*}

\paragraph{Stability of low-rank excitation subspaces.}
Figure~\ref{fig:rank1_overlap_multi_basis} examines the robustness of low-rank excitation manifolds derived from inexpensive perturbative calculations. {\color{black}
Here the subspace overlap between MP2 and CCSD rank-one manifolds is computed from the top-$r$ left singular subspaces of the reshaped $T_2$ tensor (pair-virtual $\times$ pair-occupied), using the normalized projector overlap $\mathrm{ov}(r)=\frac{1}{r}\|P^{\mathrm{MP2}}_r P^{\mathrm{CCSD}}_r\|_F^2$ (equivalently, the mean squared cosine of principal angles).} Panel~(a) shows the weighted average subspace overlap (wAUC) between MP2-derived and CCSD-derived rank-one excitation spaces, plotted against the correlation-energy discrepancy $|E_{\mathrm{corr}}^{\mathrm{CCSD}}-E_{\mathrm{corr}}^{\mathrm{MP2}}|$. Marker size indicates the number of retained rank-one operators $R_{\varepsilon}$ selected using a global singular-value threshold. The reported wAUC is a weighted average of $\mathrm{ov}(r)$ over $r\le R$, with weights chosen to emphasize the leading ranks (see Appendix~\ref{app:fixed_rank1} for the exact definition). {\color{black}
Across the benchmark set, the median wAUC remains $0.85$ or higher in all three basis sets, with the lowest-outlier cases occurring for C$_6$H$_6$ where MP2 correlation energy employing STO-3G basis set differs most strongly from CCSD.}

Despite variations in correlation energy and basis set, the dominant excitation subspaces exhibit consistently high overlap, indicating that MP2 captures the leading algebraic structure of the CCSD excitation manifold. Panel~(b) further shows that this agreement persists as a function of normalized rank fraction $r/R$, with median overlaps remaining close to unity across basis sets. Panel~(c) illustrates the growth of the retained rank-one operator count $R_{\varepsilon}$ with system size, demonstrating that the dominant excitation space grows systematically but remains a small fraction of the full $T_2$ manifold.

{\color{black}
In SW-type effective Hamiltonian constructions, the generator is chosen to suppress couplings between a model space and its complement, and practical implementations rely on truncating this generator to a structured subset of dominant channels. The observed stability of low-rank $T_2$-derived manifolds under MP2 proxies provides numerical support for such truncation/masking heuristics: it suggests that the leading couplings that would enter a low-order SW generator can be identified inexpensively and then updated across related instances (e.g., along a geometry scan) by streaming coefficient changes rather than recompiling the two-qubit fabric.}

Collectively, these results justify MP2-guided truncation and masking strategies in \texttt{COMPOSER}: a relatively small number of rank-one operators can capture the dominant excitation subspace across a broad range of molecules and basis sets. Architecturally, this stability underpins mask-aware workflows in which refinement proceeds by updating classical masks and re-dialing \texttt{PREP} amplitudes, while preserving the underlying multiplexed circuit fabric. See Appendix~\ref{app:fixed_rank1} for the MP2-guided truncation procedure and stability diagnostics underlying these masking choices.
We emphasize that this section supports a different objective than reducing asymptotic query complexity: it supports eliminating \emph{structural recompilation} across instance families by ensuring that the operator representation admits a stable, fixed-topology multiplexed realization. A detailed comparison of logical gate counts or fault-tolerant resource estimates under different masking strategies is deferred to future work.

\begin{figure}[t]
\centering
\includegraphics[width=\linewidth]{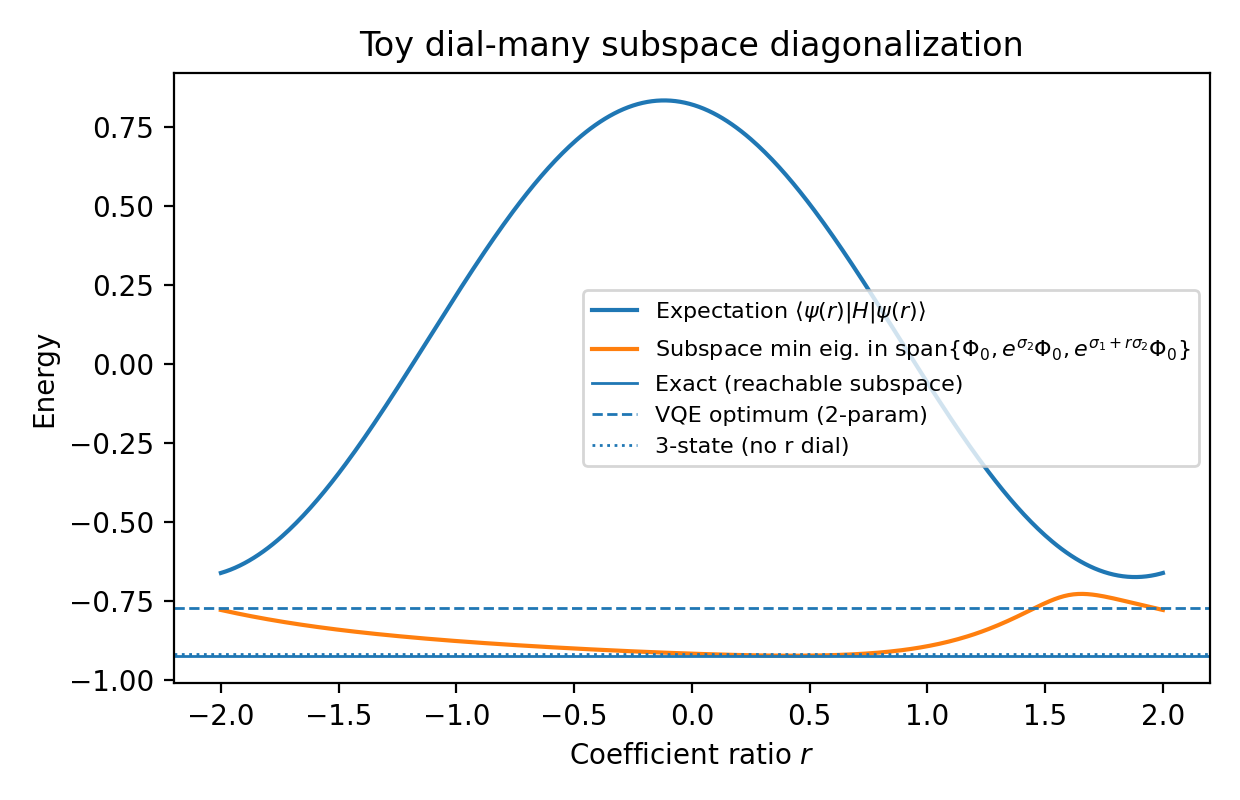}
\caption{Toy illustration of a ``compile-once, dial-many'' non-orthogonal subspace solve aligned with the \texttt{COMPOSER} execution model. A fixed state-preparation topology generates a small non-orthogonal determinant basis from a shared operator pool and supports a continuous generator-coordinate sweep by re-dialing a single classical coefficient.}
\label{fig:toy_gcim_dial_many}
\end{figure}

\paragraph{Illustrative end-to-end subspace diagonalization with a continuous generator-coordinate dial.}
Figure~\ref{fig:toy_gcim_dial_many} provides a minimal end-to-end illustration of this ``compile-once, dial-many'' pattern in a non-orthogonal subspace solve.
Here, we consider a 4-qubit, 2-electron determinant subspace generated from a reference $\ket{\Phi_0}=\ket{1100}$ by two commuting anti-Hermitian generators $(\hat\sigma_1,\hat\sigma_2)$ implemented as number-conserving Givens rotations on disjoint orbital pairs.
A conventional separable product ansatz $\ket{\psi(\theta_1,\theta_2)}=e^{\theta_2\hat\sigma_2}e^{\theta_1\hat\sigma_1}\ket{\Phi_0}$ is expressivity-limited and converges to a higher energy (here $E_{\mathrm{VQE}}=-0.774$ in arbitrary units), while a small non-orthogonal basis $\{\ket{\Phi_0},e^{\hat\sigma_1}\ket{\Phi_0},e^{\hat\sigma_2}\ket{\Phi_0}\}$ yields a much lower generalized-eigenvalue estimate ($E=-0.918$).
Finally, sweeping a continuous generator coordinate $r$ in $e^{\hat\sigma_1+r\hat\sigma_2}\ket{\Phi_0}$ within the fixed 3-state span $\{\ket{\Phi_0},e^{\hat\sigma_2}\ket{\Phi_0},e^{\hat\sigma_1+r\hat\sigma_2}\ket{\Phi_0}\}$ recovers the exact subspace ground energy (here $E_{\mathrm{exact}}=-0.924$) without changing circuit topology, illustrating the ``compile-once, dial-many'' premise in a subspace-diagonalization setting.
From the \texttt{COMPOSER} perspective, both discrete basis updates (activating/deactivating generators via masks) and continuous updates (changing coefficient ratios such as $r$) correspond to re-dialing \texttt{PREP} angles and classical selector data, while leaving the compiled two-qubit fabric unchanged.
We emphasize that the objective here is not to reduce asymptotic query complexity, but to eliminate \emph{structural recompilation} across instance families by ensuring that the operator representation admits a stable, fixed-topology multiplexed realization.

% =========================
\section{Conclusion and Outlook}

\texttt{COMPOSER} provides a practical bridge from nested low-rank tensor factorizations to hardware-native, mask-aware block encodings for molecular simulation. By fixing the multiplexing topology and encapsulating all instance dependence in single-qubit rotations, it delivers depth-optimal oracles driven by a single QSP ladder, predictable ancilla budgets comprising only selector and signal registers, and zero recompilation overhead across geometry sweeps, active-space growth, and adaptive similarity transformations. These properties are supported by numerical evidence demonstrating near-linear scaling of rank-one Hamiltonian factorizations and robustness of low-rank excitation subspaces, which validate the structural assumptions of the architecture.

{\color{black}
Beyond downfolded Hamiltonian simulation, the same mask-aware interface is immediately applicable to non-orthogonal subspace solvers that require repeated evaluation of Hamiltonian and overlap matrix elements in an expanding basis (e.g., QSE/NOQE/GCIM), since basis-state updates can be expressed as mask/parameter updates on a fixed oracle template~\cite{McClean2017QSE,Huggins2020NonOrthogonalVQE,Baek2023NOQE,Zheng2023QuantumGCM,Zheng2024UnleashedGCIM}.
This offers a concrete end-to-end application in which the compile-once principle directly mitigates the otherwise quadratic growth in distinct compiled circuits as the subspace dimension increases.
}

Methodologically, deterministic, number-conserving ladders provide a unified implementation of bilinear, pair-excitation, and projected-quadratic rank-one operators, enabling efficient block encodings for both Hamiltonians and anti-Hermitian generators. Practically, the mask-aware ``compile-once, dial-later'' execution model reduces end-to-end latency, simplifies resource estimation, and decouples adaptive algorithm design from circuit recompilation—features that are particularly attractive for NISQ demonstrations and early fault-tolerant deployments.

We emphasize that the present work is architectural rather than competitive in asymptotic resource counts. Detailed logical Toffoli estimates and fault-tolerant overhead comparisons under specific hardware models are natural next steps, but are orthogonal to the compile-once principle developed here. The numerical results validate the structural premises (low-rank operator growth and excitation-subspace stability) that justify freezing the logical circuit topology, rather than demonstrating an end-to-end runtime advantage over existing simulation frameworks.

Looking ahead, natural extensions of this work include automated hardware-aware scheduling of two-qubit gates for large instances on heavy-hex and grid-based architectures, systematic studies of mask-selection strategies and error-budget allocation, and generalizations beyond electronic structure to lattice models and open-system generators. 
{\color{black} Angle-selection and refinement (Sec.~\ref{sec:composer} and Introduction) remain important directions, including systematic improvements beyond CCSD-seeded doubles via tighter factorizations, expanded operator pools, or hybrid diagnostic-guided updates—without altering the fixed two-qubit topology.}
More broadly, the connections between projected-quadratic ladders, local Jastrow factors, and tensor-network operator representations suggest a unifying perspective in which fixed-topology, data-driven quantum circuits capture correlation efficiently across a wide range of quantum simulation settings.

\begin{acknowledgments}
B.P. acknowledges the support from the Early Career Research Program by the U.S. Department of Energy, Office of Science, under Grant No. FWP 83466.  Y.L. acknowledges the support by the U.S. Department of Energy, Office of Science, Advanced Scientific Computing Research, under contract number DE-SC0025384. K.K. acknowledges the Quantum Algorithms and Architecture for Domain Science Initiative (QuAADS), a Laboratory Directed Research and Development (LDRD) program at PNNL.

\end{acknowledgments}

%%%%%%%%%%%%%%%%%%%%%%%%%%%%%%%%%%%%%
\appendix

\section{Deterministic Ladder Constructions and Hardware Realizations}
\label{app:ladders}

This appendix collects implementation-level details for the deterministic ladder circuits used in Secs.~\ref{subsec:prep_u_state} and~\ref{subsec:prep_two_electron}. In the main text we emphasize the existence, fixed topology, and number-conserving nature of these ladders; here we record explicit angle recursions and convenient product decompositions that can be used directly in code generation. {\color{black}
Throughout this appendix we use the convention $\prod_{k=1}^{m} U_k := U_m\cdots U_2 U_1$ (rightmost factor acts first).}

\subsection{One-electron ladder: recursion and explicit product form}

Fix an ordering $(p_1,p_2,\dots,p_{n-1})$ of non-pivot orbitals (with pivot $r$). The ladder rotation angles may be chosen recursively.
{\color{black}
Define tail norms $s_k:=\sqrt{|u_r|^2+\sum_{j\ge k}|u_{p_j}|^2}$ for $k=1,\dots,n-1$, and set $s_n:=|u_r|$. 
A convenient deterministic choice of ladder angles is then
\begin{align}
\theta_{p_k r}
&=
\arctan\left(\frac{|u_{p_k}|}{s_{k+1}}\right) \notag \\
&=
\arctan\left(
\frac{|u_{p_k}|}{\sqrt{|u_r|^2+\sum_{j> k}|u_{p_j}|^2}}
\right),
\label{eq:theta_recursion}
\end{align}
equivalently $\sin\theta_{p_k r}=|u_{p_k}|/s_k$ and $\cos\theta_{p_k r}=s_{k+1}/s_k$.}
so that amplitude is transferred sequentially from $\lvert r\rangle$ to the desired orbitals with correct magnitudes. {\color{black}
Residual complex phases are restored by diagonal phase shifts on orbitals; since the overall many-body phase is irrelevant we may fix the gauge so that $u_r\in\mathbb{R}_{\ge 0}$ and apply $R_z^{(p)}(\phi_p)$ with $\phi_p=\arg(u_p)$ for $p\neq r$ (or include $p=r$ as well, which changes only a global phase).}

Collecting all operations yields the preparation-form ladder unitary.
{\color{black}
Define a real-amplitude ladder $V_u:=\prod_{k=1}^{n-1} G_{p_k r}(\theta_{p_k r})$ and a diagonal phase operator 
$D_u:=\prod_{p\neq r} R_z^{(p)}(\arg u_p)$. 
Then the preparation-form ladder may be written compactly as
\begin{align}
U^{(u)}_{\texttt{prep}}
=
D_u\,V_u\,X_r ,
\label{eq:Uu_explicit}
\end{align}
which satisfies $U^{(u)}_{\texttt{prep}}\lvert 0^n\rangle=\lvert u\rangle$.}
The inverse is obtained by reversing the gate order and negating all angles.

\subsection{Two-electron ladder: recursion and explicit product form}

Fix an ordering of the non-pivot unordered pairs $(p_1,q_1),\dots,(p_m,q_m)$ with $m=\binom{n}{2}-1$ (pivot pair $(r,s)$). The ladder angles and phases may be chosen as follows.
{\color{black}
Define tail norms $s_k:=\sqrt{|u_{rs}|^2+\sum_{j\ge k}|u_{p_j q_j}|^2}$ for $k=1,\dots,m$, and set $s_{m+1}:=|u_{rs}|$.
A convenient deterministic choice is
\begin{align}
\theta_k
&=
\arctan\left(\frac{|u_{p_k q_k}|}{s_{k+1}}\right) \notag \\
&=
\arctan\left(
\frac{|u_{p_k q_k}|}{\sqrt{|u_{rs}|^2+\sum_{j> k}|u_{p_j q_j}|^2}}
\right),
\label{eq:Theta_recursion_phased}
\end{align}
equivalently $\sin\theta_k=|u_{p_k q_k}|/s_k$ and $\cos\theta_k=s_{k+1}/s_k$.}
Here, amplitude and phase are transferred sequentially from the pivot pair to the desired determinants. {\color{black}
$\phi_k = \arg(u_{p_k q_k})$ fixes the overall phase so that $u_{rs}\in\mathbb{R}_{\ge 0}$.} The corresponding preparation-form ladder is
\begin{align}
U^{(u)}_{\texttt{prep}}
=
\Big(
\prod_{k=1}^{m}
G_{p_k q_k,rs}(\theta_k,\phi_k)
\Big)
X_r X_s ,
\label{eq:UC_explicit_phased}
\end{align}
which satisfies $U^{(u)}_{\texttt{prep}}\lvert 0^n\rangle=\lvert u\rangle$. 

%%%%%%%%%%%%%%%%%%%%%%%%%%%%%%%%%%

\subsection{Two-Electron Givens--SWAP Networks and Hardware Mapping}\label{app:A.3}

This appendix details the routing and scheduling of the two-electron deterministic ladders introduced in Sec.~\ref{subsec:prep_two_electron}, with emphasis on SWAP overheads and connectivity-dependent depth scaling.~\cite{Kivlichan2018FermionicSwap,Babbush2018LowDepth,Cowtan2020QubitRouting,Sivarajah2020tket,Zulehner2018MappingIBM}

{\color{black}
The two-electron preparation circuit may be expressed directly in terms of the phased pair--Givens blocks from Eq.~\eqref{eq:pair_Givens_phased},
\begin{align}
U^{(u)}_{\texttt{prep}}
=
\Bigl(
\prod_{k=1}^{m}
G_{p_k q_k,rs}(\theta_k,\phi_k)
\Bigr)
X_r X_s .
\label{eq:B_prep}
\end{align}
In practice, the phase parameter $\phi_k$ is realized by a constant-depth pattern of local $Z$ rotations within the fixed four-qubit block implementing $G_{p_k q_k,rs}$, which enables consolidation of many diagonal phases across neighboring blocks.}
%contains $\mathcal O(n^2)$ pair-Givens rotations and $2\mathcal O(n^2)$ single-qubit phase gates in total. 
Below we summarize routing strategies that minimize depth under common hardware connectivity assumptions.

\paragraph{Minimal SWAP schedule}
During the ladder execution, the pivot pair $(r,s)$ must remain adjacent.
A single-pivot routing scheme suffices:
\begin{align}
(r,s,q_1,q_2)
&\xrightarrow[(r,q_1)]{\texttt{fSWAP}}
(q_1,s,r,q_2) \notag\\
&\xrightarrow[(s,q_2)]{\texttt{fSWAP}}
(q_1,q_2,r,s),
\end{align}
requiring two long-range fermionic SWAPs per four-qubit block.
No additional local SWAPs are necessary.

\begin{table*}
\centering
\caption{Depth and entangling-gate counts for one four-qubit Givens block
under representative connectivity graphs.}
\label{tab:B_network}
\begin{tabular}{lccc}
\hline
Architecture &
Routing cost &
CZ gates per block &
Depth scaling \\ \hline
Linear / heavy-hex
& $(4d_{\mathrm{g}}-2)$ \texttt{fSWAP}s
& $8+6(4d_{\mathrm{g}}-2)$
& $\mathcal O(n^2)$ \\
2D square grid
& $(4\ell-4)$ \texttt{fSWAP}s
& $8+6(4\ell-4)$
& $\mathcal O(n)$ \\
All-to-all
& none
& $8$
& $\mathcal O(n)$ \\ \hline
\end{tabular}
\end{table*}

\paragraph{Topology-dependent resource scaling} The depth and entangling-gate counts for one four-qubit Givens block under representative connectivity graphs are given in Tab.~\ref{tab:B_network}. Here $d_{\mathrm{g}}$ denotes the heavy-hex Manhattan distance and $\ell$ the linear dimension of the square grid. On linear or heavy-hex lattices, the pair-Givens blocks must be executed serially, leading to quadratic depth. Two-dimensional connectivity allows parallel scheduling of non-overlapping blocks, reducing depth to $\mathcal O(n)$, while all-to-all connectivity eliminates routing overhead entirely.

{\color{black}
\paragraph{Single-qubit phase consolidation}
Many diagonal $Z$ rotations introduced when compiling the phased pair--Givens blocks can be commuted and merged into neighboring blocks (or absorbed into the internal phase parameters of adjacent decompositions), so the \emph{visible} single-qubit \emph{depth} contribution can often be reduced substantially even when the \emph{total} number of single-qubit rotations scales as $\mathcal{O}(n^2)$. The exact consolidation achieved depends on the chosen decomposition of $G_{pq,rs}(\theta,\phi)$ and the scheduling/routing pattern.}

%%%%%%%%%%%%%%%%%%%%%%%%%%%%%%%%%%%%%

\subsection{Preparation versus number-conserving ladder forms}
\label{app:prep_vs_nc}

The ladder constructions in Secs.~\ref{subsec:prep_u_state} and~\ref{subsec:prep_two_electron} admit two closely related realizations, depending on whether the goal is state preparation from the vacuum or basis rotation within an existing many-electron state. Both realizations are generated by the same underlying sequence of Givens rotations and differ only in whether particle injection is included.

\paragraph{Preparation form.}
When the system register is initialized in the vacuum $\lvert 0^n\rangle$, the ladder unitaries may be preceded by explicit particle injection using Pauli-$X$ gates on selected pivot orbitals (or pivot pairs),
\begin{align}
U^{(u)}_{\texttt{prep}} = U_u X_r
~~ \text{or} ~~
U^{(u)}_{\texttt{prep}} = U_u X_r X_s ,
\label{eq:U_u_prep}
\end{align}
for one- and two-electron ladders, respectively. In this form, the circuit prepares the desired few-electron state deterministically from the vacuum and is useful for initializing reference states or illustrating ladder behavior.

\paragraph{Number-conserving form.}
Omitting the initial particle-injection step yields the strictly number-conserving ladder unitary $U_u$. This form acts as a basis rotation within a fixed particle-number sector and is agnostic to the specific many-electron state residing on the system register. Crucially, the numerical values of all rotation angles and phases are identical to those used in the preparation form; only the initial excitation is absent.

\paragraph{Role in \texttt{COMPOSER}.}
In the \texttt{COMPOSER} architecture, the number-conserving form $U_u$ is used inside rank-one block-encoding adaptors. This choice ensures that ladder circuits can be applied to arbitrary system states without altering particle number, enabling reuse of a fixed circuit topology across geometry changes, active-space growth, and masked similarity transformations. Particle injection via pivot $X$ gates is therefore a convenience for state preparation, not a requirement for block encoding.

Throughout the remainder of this work, $U_u$ denotes the number-conserving ladder form unless state preparation from the vacuum is stated explicitly.

%%%%%%%%%%%%%%%%%%%%%%%%%%%%%%%%%%%%%

\section{Proofs and implementation details for rank-one block encodings}
\label{app:rank1_proofs}

This appendix collects proofs and low-level circuit constructions for the rank-one block-encoding primitives introduced in Sec.~\ref{subsubsec:uv_block}. The derivations collected are provided for completeness; they do not introduce new primitives beyond standard block-encoding and LCU constructions, but establish correctness and normalization conventions used in the main text. The main text focuses on the resulting adaptor interfaces and how they assemble into a compile-once \texttt{PREP-SELECT-PREP}$^\dagger$ architecture; the derivations are deferred here for readability.

\subsection{Proof of Lemma~\ref{lem:rank1_block}}

\begin{figure*}
    \centering
    \begin{quantikz}[row sep=0.35cm, column sep=0.5cm]
    \lstick{$\ket{0}$ (operator anc.)}      
    & \qw 
    & \qw 
    & \qw\gategroup[2,steps=6,style={dashed,rounded corners,inner sep=10pt},background]{$W_s$} 
    & \gate{H_A}\gategroup[2,steps=3,style={dashed,rounded corners,fill=blue!20, inner xsep=2pt},background,label style={label position=below,anchor=north, xshift=1cm, yshift=0.3cm}]{$U_{\Pi_0}$} 
    & \ctrl{1}
    & \gate{H_A}
    & \gate{R_y(\theta)} 
    & \qw
    & \meter[label style={label position=below,anchor=north,xshift=0.7cm}]{=0} \\     
    \lstick{$\ket{\psi}$\,(system)}
    & \qwbundle{n} 
    & \qw 
    & \gate{U^{(v)\dagger}_{\texttt{prep}}}
    & \qw 
    & \gate{-R_0}
    & \qw 
    & \gate{U^{(u)}_{\texttt{prep}}}
    & \qw
    &~~\hat L_s/\alpha_s \ket{\psi}
    \end{quantikz}
    \caption{Sketch of an $(\alpha_s,1,\epsilon_s)$ block encoding of the dyad $\hat L_s=\lambda_s\ket{u}\bra{v}$ on the one-excitation subspace, with {\color{black}$\cos(\theta/2)=|\lambda_s|/\alpha_s$} (and an optional phase gate to encode $\arg\lambda_s$ if $\lambda_s$ is not chosen real). 
    {\color{black}The shaded gadget $U_{\Pi_0}$ is the standard single-ancilla LCU construction of the vacuum projector $\Pi_0=\ket{0^n}\bra{0^n}$ from the vacuum reflection $R_0=I-2\Pi_0$: $U_{\Pi_0}=(H_A\otimes I)\big(\ket{0}\bra{0}\otimes I+\ket{1}\bra{1}\otimes(-R_0)\big)(H_A\otimes I)$.}
    Conjugation by ${\color{black}U^{(v)}_{\texttt{prep}}}$ and ${\color{black}U^{(u)}_{\texttt{prep}}}$ yields $(\bra{0}\otimes I)\,W_s\,(\ket{0}\otimes I)=\ket{u}\bra{v}$.}
    \label{fig:circuit_rank1}
\end{figure*}
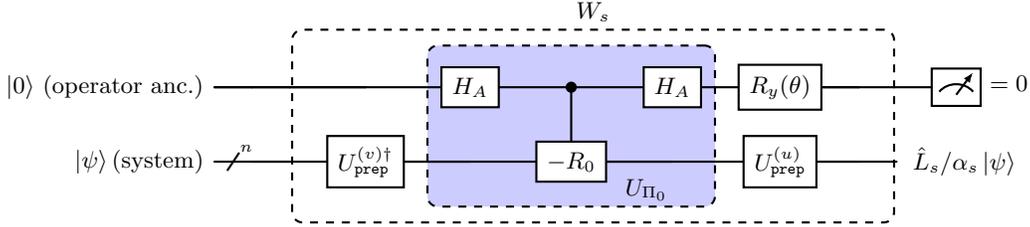

\begin{proof}
Since $\hat a_u^\dagger$ and $\hat a_v^\dagger$ are linear combinations of canonical creation operators, applying them to the vacuum produces normalized \emph{one-electron} states (i.e., superpositions of computational-basis determinants with a single occupied orbital):
\begin{align}
\ket{u}:=\hat a_u^\dagger\ket{0^n},
~~
\ket{v}:=\hat a_v^\dagger\ket{0^n}.
\end{align}
Hence the bilinear rank-one operator acts on the \emph{one-electron} sector as
\begin{align}
\hat L_s=\lambda_s\,\hat a_u^\dagger \hat a_v=\lambda_s\,\ket{u}\bra{v}.
\end{align}
Let $U^{(u)}_{\texttt{prep}}$ and $U^{(v)}_{\texttt{prep}}$ be the deterministic one-electron preparation circuits from Sec.~\ref{subsec:prep_u_state} satisfying
\begin{align}
U^{(u)}_{\texttt{prep}}\ket{0^n}=\ket{u},
~~
U^{(v)}_{\texttt{prep}}\ket{0^n}=\ket{v}.
\end{align}
Denote by $\Pi_0:=\ket{0^n}\bra{0^n}$ the projector onto the system vacuum.

Introduce a single ancilla qubit $A$ initialized in $\ket{0}_A$ and define the vacuum reflection
\begin{align}
{\color{black}
R_0 := I_S - 2\Pi_0, \qquad \Pi_0:=\ket{0^n}\bra{0^n}.}
\end{align}
{\color{black}
The unitary $R_0$ can be implemented as a multi-controlled phase flip on $\ket{0^n}$ (equivalently, conjugating a phase flip on $\ket{1^n}$ by $X$ gates on all system qubits), as indicated schematically in Figure~\ref{fig:circuit_rank1}.}

{\color{black}
We first construct a single-ancilla \emph{deterministic} block encoding of $\Pi_0$ using a two-branch LCU over $\{I_S,-R_0\}$. Define
\begin{widetext}
\begin{align}
U_{\Pi_0}
:=
(H_A\otimes I_S)\,
\Big(\ket{0}\bra{0}_A\otimes I_S + \ket{1}\bra{1}_A\otimes (-R_0)\Big)\,
(H_A\otimes I_S).
\end{align}    
\end{widetext}
A direct calculation gives
\begin{align}
(\bra{0}_A\otimes I_S)\,U_{\Pi_0}\,(\ket{0}_A\otimes I_S)
=
\frac{1}{2}\bigl(I_S- R_0\bigr)
=
\Pi_0,
\end{align}
so $U_{\Pi_0}$ is a $(1,1,0)$ block encoding of $\Pi_0$.}

{\color{black}
Now define
\begin{align}
\widetilde W_s
:=
(I_A\otimes U^{(u)}_{\texttt{prep}})\,
U_{\Pi_0}\,
(I_A\otimes U^{(v)\dagger}_{\texttt{prep}}).
\end{align}
Projecting the ancilla yields
\begin{align}
(\bra{0}_A\otimes I_S)\,\widetilde W_s\,(\ket{0}_A\otimes I_S)
&= U^{(u)}_{\texttt{prep}}\,\Pi_0\,U^{(v)\dagger}_{\texttt{prep}} \notag \\
&= \ket{u}\bra{v}.
\end{align}
Thus $\widetilde W_s$ is a $(1,1,0)$ block encoding of $\ket{u}\bra{v}$.}

{\color{black}
Finally, since $\hat L_s=\lambda_s\ket{u}\bra{v}$ has $\|\hat L_s\|=|\lambda_s|$, we may choose $\alpha_s:=|\lambda_s|$ and (without loss of generality) absorb $\arg(\lambda_s)$ into $u$ or $v$ so that $\lambda_s\in\mathbb R_{\ge 0}$. With this choice,
\begin{align}
(\bra{0}_A\otimes I_S)\,\widetilde W_s\,(\ket{0}_A\otimes I_S)
=
\frac{\hat L_s}{\alpha_s}.
\end{align}
If a larger normalization $\alpha\ge \alpha_s$ is desired for later multiplexing, it can be handled at the outer LCU layer (i.e., by rescaling coefficients in the selector preparation), without changing the single-ancilla adaptor topology.}

The above equalities hold for ideal unitaries. In practice, $U^{(u)}_{\texttt{prep}}$, $U^{(v)}_{\texttt{prep}}$, and the multi-controlled implementation of $R_0$ are synthesized to finite precision, yielding an additive block-encoding error $\Delta_s$ with $\|\Delta_s\|\le \epsilon_s$.
\end{proof}

\subsection{Proof of Lemma~\ref{lem:quad_rank1}}

\begin{proof}
We work in the rotated basis defined by $U^{(\mu)}$ (applied before the gadget and inverted after). Introduce three ancilla registers: (i) an index register $I$ encoding $\xi\in\{1,\dots,R_\mu\}$, (ii) a flag qubit $f$, and (iii) a QSVT signal qubit $g$. {\color{black} All ancillas are initialized in $\ket{0}$.}

\paragraph{\texttt{PREP} (index loading).}
{\color{black}
Prepare $I$ with amplitudes proportional to $\sqrt{|\lambda^{(\mu)}_\xi|}$:}
\begin{align}
U^{(\mu)}_{\texttt{prep}}\ket{0}_{I}
=
\frac{1}{\sqrt{\Lambda_\mu}}\sum_{\xi=1}^{R_\mu}
\sqrt{\bigl|\lambda^{(\mu)}_{\xi}\bigr|} \,\ket{\xi}_{I}.
\label{eq:prep-index-O}
\end{align}
This can be synthesized by a unary Givens ladder (depth $\mathcal O(R_\mu)$) or a standard binary state-preparation routine (depth $\mathrm{polylog}(R_\mu)$ given classical data access).

\paragraph{\texttt{SELECT} (deterministic block encoding of $\hat n_{\mu\xi}$).}
{\color{black}
For each $\xi$, define a two-qubit unitary acting on the rotated system mode $\xi$ and the flag qubit $f$:
\begin{align}
W_{\mu\xi} := X_f\,\mathrm{CNOT}_{\xi\rightarrow f},
\end{align}
where the control is the system qubit representing occupation of mode $\xi$ \emph{in the rotated basis}. A direct basis-state check shows that
\begin{align}
(\bra{0}_f\otimes I)\,W_{\mu\xi}\,(\ket{0}_f\otimes I) = \hat n_{\mu\xi}.
\label{eq:block-nxi}
\end{align}
Thus $W_{\mu\xi}$ is a deterministic $(1,1,0)$ block encoding of $\hat n_{\mu\xi}$.}

{\color{black}
Now define the multiplexed \texttt{SELECT} operator}
\begin{align}
U^{(\mu)}_{\texttt{sel}}
=
\sum_{\xi=1}^{R_\mu}\ket{\xi}\bra{\xi}_{I}\otimes
\Bigl(e^{i\phi_\xi}\,W_{\mu\xi}\Bigr),
\label{eq:select-nxi}
\end{align}
{\color{black}
where $\phi_\xi=\arg(\lambda^{(\mu)}_\xi)$ (for real eigenvalues, $\phi_\xi\in\{0,\pi\}$), so that
$e^{i\phi_\xi}\sqrt{|\lambda_\xi|}\sqrt{|\lambda_\xi|}=\lambda_\xi$.}

\paragraph{LCU block (encoding $\hat O_\mu$).}
Define the \texttt{PREP-SELECT-PREP}$^\dagger$ sandwich
\begin{align}
U_{\hat O_\mu}
:=
U^{(\mu)\dagger}_{\texttt{prep}}\,
U^{(\mu)}_{\texttt{sel}}\,
U^{(\mu)}_{\texttt{prep}}.
\end{align}
Projecting the ancillas $(I,f)$ onto $\ket{0_I0_f}$ gives
\begin{align}
&(\bra{0}_I\bra{0}_f\otimes I)\,U_{\hat O_\mu}\,(\ket{0}_I\ket{0}_f\otimes I) \notag \\
&\qquad=
\frac{1}{\Gamma_\mu}\sum_{\xi=1}^{R_\mu}
\bigl|\lambda^{(\mu)}_{\xi}\bigr|\,
e^{i\phi_\xi}\,
\hat n_{\mu\xi}
\notag\\
&\qquad=
\frac{1}{\Gamma_\mu}\sum_{\xi=1}^{R_\mu}
\lambda^{(\mu)}_{\xi}\,
\hat n_{\mu\xi} \notag \\
&=
\hat O_\mu/\Gamma_\mu,
\end{align}
which is a deterministic $(\Gamma_\mu,\,\mathfrak{a}_I+1,\,0)$ block encoding of $\hat O_\mu$.

{\color{black}
\paragraph{From $\hat O_\mu$ to $\hat O_\mu^2$ (fixed degree-2 QSVT).}
Since $\hat O_\mu$ is Hermitian and $\|\hat O_\mu\|\le \Gamma_\mu$, the above gives a valid block encoding of $\hat O_\mu/\Gamma_\mu$. Apply QSVT/QSP with the \emph{exact} polynomial $P_2(x)=x^2$ (degree $2$) to this block encoding. That is, there exists a fixed phase list $\boldsymbol{\varphi}^{(2)}=(\varphi_0,\varphi_1,\varphi_2)$ such that the QSVT circuit
\begin{align}
U_{\hat O_\mu^2}
:=
\mathrm{QSVT}\!\big(P_2,\,U_{\hat O_\mu};\,g\big)
\end{align}
satisfies
\begin{align}
&(\bra{0}_g\bra{0}_I\bra{0}_f\otimes I)\,U_{\hat O_\mu^2}\,(\ket{0}_g\ket{0}_I\ket{0}_f\otimes I)\notag \\
&\qquad=
(\hat O_\mu/\Gamma_\mu)^2
=
\hat O_\mu^2/\Gamma_\mu^2.
\end{align}
Because $P_2$ is implemented exactly, this introduces \emph{no polynomial-approximation error}; the only error source is gate synthesis if angles are approximated.}

\paragraph{Resources.}
The \texttt{PREP} stage costs $\mathcal O(R_\mu)$ two-qubit gates in a unary ladder and
$\mathrm{polylog}(R_\mu)$ depth in standard binary state preparation.
The \texttt{SELECT} stage applies a single two-qubit gate $\mathrm{CNOT}_{\xi\rightarrow f}$ (plus $X_f$)
conditioned on $\xi$; in unary this is $\mathcal O(R_\mu)$ controlled interactions, and in binary it incurs only
polylogarithmic decoding overhead on top of $\mathcal O(R_\mu)$ controls.
The basis rotations $U^{(\mu)}$ and $U^{(\mu)\dagger}$ (if included explicitly) cost $\mathcal O(n)$ two-qubit
Givens rotations.
{\color{black}The degree-2 QSVT step adds a constant overhead: two uses of $U_{\hat O_\mu}$ (and/or $U_{\hat O_\mu}^\dagger$) and $\mathcal O(1)$ single-qubit rotations on the signal qubit $g$.}
{\color{black}
Ancilla width is $\mathfrak{a}_I$ (index) $+1$ (flag) $+1$ (QSVT signal), with all ancillas initialized and postselected in $\ket{0}$.}
\end{proof}

\subsection{Proof of Theorem~\ref{thm:sum_block}}

\begin{proof}
For each $s$, by hypothesis,
\[
(\bra{0^t}\otimes I_S)\,W_s\,(\ket{0^t}\otimes I_S)
=
\hat L_s/\alpha_s+\Delta_s,
~~
\|\Delta_s\|\le\epsilon_s.
\]
Write $\Omega_s=|\Omega_s|e^{i\phi_s}$ and define the \texttt{SELECT} operator
\begin{align}
W_{\texttt{sel}}
=
\sum_{s=1}^{\ell_H}\ket{s}\bra{s}\otimes \big(e^{i\phi_s}W_s\big),
\label{eq:W_sel}
\end{align}
which applies $e^{i\phi_s}W_s$ conditioned on the selector state $\ket{s}$.
Prepare the selector register using
\begin{align}
U_{\texttt{prep}}\ket{0^{a}}
=
\frac{1}{\sqrt{\alpha}}
\sum_{s=1}^{\ell_H}\sqrt{|\Omega_s|\,\alpha_s}\,\ket{s},
~~
\alpha=\sum_{s=1}^{\ell_H}|\Omega_s|\,\alpha_s.
\end{align}
Define
\[
W=(U_{\texttt{prep}}^\dagger\otimes I)\,W_{\texttt{sel}}\,(U_{\texttt{prep}}\otimes I).
\]
Projecting the $\mathfrak{a}+t$ ancillas onto $\ket{0^{\mathfrak{a}+t}}$ yields
\begin{align}
&(\bra{0^{\mathfrak{a}+t}}\otimes I)\,W\,(\ket{0^{\mathfrak{a}+t}}\otimes I) \notag \\
&~~=
\frac{1}{\alpha}
\sum_{s}|\Omega_s|\,\alpha_s\,
(\bra{0^t}\otimes I)\,e^{i\phi_s}W_s\,(\ket{0^t}\otimes I)
\notag\\
&~~=
\frac{1}{\alpha}
\sum_{s}|\Omega_s|e^{i\phi_s}
\bigl(\hat L_s+\alpha_s\Delta_s\bigr)
\notag\\
&~~=
\frac{\hat H}{\alpha}+\Delta,
\end{align}
where
\[
\Delta
=
\frac{1}{\alpha}
\sum_{s}|\Omega_s|\,\alpha_s\,e^{i\phi_s}\Delta_s,
~~
\|\Delta\|\le\epsilon_{\mathrm{LCU}}.
\]
The depth bound follows from the cost of the selector-controlled application
of $W_s$ together with the two invocations of the state-preparation unitary
$U_{\texttt{prep}}$.
\end{proof}

%%%%%%%%%%%%%%%%%%%%%%%%%%%%%%%%%%%%%

\section{Quantum Signal Processing for Exponentiating Block-Encoded Generators}
\label{app:qsp_details}

This appendix summarizes the quantum signal processing (QSP) construction used in Sec.~\ref{subsubsec:pqr1_exp_block} to implement the exponential of a block-encoded anti-Hermitian generator. The material here is standard and included for completeness; see Refs.~\cite{LowChuang2017QSP,Gilyen2019QSVT} for full treatments.

\subsection{Polynomial approximation}

Let $U$ be a block encoding of a Hermitian operator $\hat{\mathbb{A}}_s$ with normalization $\alpha$, i.e.
\begin{align}
(\bra{0}\otimes I)\,U\,(\ket{0}\otimes I) = \hat{\mathbb{A}}_s/\alpha ,
~~
\|\hat{\mathbb{A}}_s\|\le \alpha .
\end{align}
To approximate the unitary $e^{-i\hat{\mathbb{A}}_s}$, QSP constructs a polynomial $P_d(x)$ such that
\begin{align}
\sup_{x\in[-1,1]} \bigl| P_d(x) - e^{-i\alpha x} \bigr| \le \varepsilon .
\end{align}

A convenient choice is a truncated Chebyshev expansion
\begin{align}
P_d(x) = \sum_{k=0}^{d} c_k T_k(x),
\end{align}
where $T_k(x)$ denotes the Chebyshev polynomial of the first kind. For $\alpha>0$, one may choose an even degree
\begin{align}
d = \mathcal{O}\big(\alpha + \log(1/\varepsilon)\big), \label{eq:app_C_d}
\end{align}
which guarantees uniform approximation error $\varepsilon$ on $[-1,1]$. This scaling is asymptotically optimal.

\subsection{QSP implementation}

Given the polynomial $P_d(x)$, quantum signal processing realizes $P_d(\hat{\mathbb{A}}_s/\alpha)$ using a single ancilla qubit and a sequence of phase rotations.
Specifically, there exists a phase list
\(
\boldsymbol{\phi} = (\phi_0,\phi_1,\ldots,\phi_d)
\)
such that the unitary
\begin{align}
Q(\boldsymbol{\phi},U)
=
e^{i\phi_0 Z}
\prod_{k=1}^{d}
\Bigl(
U\,e^{i\phi_k Z}
\Bigr)
\end{align}
satisfies
\begin{align}
(\bra{0}\otimes I)\,Q(\boldsymbol{\phi},U)\,(\ket{0}\otimes I)
=
P_d(\hat{\mathbb{A}}_s/\alpha).
\end{align}
{\color{black}
The phases $\{\phi_k\}$ depend only on the target function $e^{-i\alpha x}$ and the approximation tolerance $\varepsilon$; they are independent of the internal structure of $\hat{\mathbb{A}}_s$. Consequently, for a \emph{fixed} normalization $\alpha$ the phase list may be computed once classically and reused across different problem instances. If $\alpha$ changes (e.g., under masking), one either recomputes the phase list or fixes a global worst-case normalization and pads the LCU with a null branch so that $\alpha$ remains invariant.}

The QSP circuit uses:
\begin{itemize}
\item one signal ancilla qubit,
\item $d$ controlled applications of the block encoding $U$,
\item $d+1$ single-qubit $Z$ rotations.
\end{itemize}

\subsection{Error propagation}

In practice, the block encoding $U$ is approximate:
\begin{align}
(\bra{0}\otimes I)\,U\,(\ket{0}\otimes I) = \hat{\mathbb{A}}_s/\alpha + \Delta,
~~
\|\Delta\|\le \epsilon' .
\end{align}
Applying QSP to such an imperfect block encoding yields
\begin{align}
(\bra{0}\otimes I)\,Q(\boldsymbol{\phi},U)\,(\ket{0}\otimes I)
=
e^{-i\hat{\mathbb{A}}_s} + \epsilon''(\epsilon',\varepsilon),
\end{align}
where the total operator-norm error satisfies
\begin{align}
\epsilon'' = \mathcal O(d\epsilon' + \varepsilon).
\end{align}
{\color{black}
where the factor of $d$ reflects the $d$ uses of the imperfect block encoding within the QSP sequence (see, e.g., the stability bounds in Refs.~\cite{LowChuang2017QSP,Gilyen2019QSVT}).}
Thus, the overall accuracy is controlled by balancing the block-encoding error $\epsilon'$ and the polynomial approximation error $\varepsilon$. In Sec.~\ref{subsubsec:pqr1_exp_block}, $\epsilon'$ arises from the LCU construction of the generator block encoding, while $\varepsilon$ is set by the chosen polynomial degree $d$.

\subsection{Application to the generator $\hat{\sigma}$}

In the main text, QSP is applied to the block encoding $U_{\hat\sigma}$ of the anti-Hermitian generator \(
\hat\sigma=\sum_s \omega_s \hat A_s
\)(equivalently, to the Hermitian operator $\hat{\mathbb A}:=i\hat\sigma$ via $e^{\hat\sigma}=e^{-i\hat{\mathbb A}}$).
{\color{black}
Crucially, a \emph{single} QSP ladder suffices to implement $e^{\hat\sigma}$, and its degree depends on the chosen normalization (e.g., $\alpha'$ or a fixed global $\bar\alpha'$) rather than directly on the number of retained rank-one terms after masking.}

%%%%%%%%%%%%%%%%%%%%%%%%%%%%%%%%

\section{Compile--Once Architecture and Error--Depth Tradeoffs}
\label{app:composer_impl}

This appendix provides additional detail underlying the compile--once execution
model introduced in Sec.~\ref{sec:composer}. The material here concerns circuit
control overheads, selector--controlled depth scaling, and the relationship
between error budgets and overall runtime. These details are not required to
understand the high-level \texttt{COMPOSER} architecture but may be useful for concrete
implementations and resource estimation.

\subsection{Selector-controlled adaptor overhead}

In the binary-multiplexed constructions of Sec.~\ref{sec:block_encoding}, each rank-one adaptor $W_s$ is executed under control of a selector register. As a result, the relevant depth contribution is that of the \emph{controlled} implementation of $W_s$, rather than the bare adaptor itself.

For the bilinear rank-one adaptors of Lemma~\ref{lem:rank1_block}, the underlying circuit consists of $\mathcal O(n)$ two-qubit Givens rotations acting on the system register and a single signal ancilla. Standard constructions allow each two-qubit gate to be promoted to a selector-controlled version with constant factor overhead, yielding selector-controlled depth $\mathcal O(n)$.

For the diagonalized Cholesky-channel adaptors of Lemma~\ref{lem:quad_rank1}, the controlled implementation additionally involves index-dependent phase kickback operations. In this case, the selector-controlled depth scales as $\mathcal O(n+R_\mu)$, where $R_\mu$ is the number of retained rotated modes in channel $\mu$. These overheads are absorbed into the depth estimates used in Sec.~\ref{subsec:rank1_resources}.

\subsection{Depth scaling versus error tolerance}

The total circuit depth of a similarity-sandwiched oracle $W_{\mathrm{eff}}^{(m)}$ depends parametrically on the total error budget $\epsilon_{\mathrm{tot}}$ through the quantum signal processing (QSP) degree $d$ as in Eq.~\eqref{eq:app_C_d}. 
Combining this with the selector-controlled depth of the generator block encoding yields the parametric scaling
\begin{widetext}
\begin{align}
\mathrm{depth}\!\left(W_{\mathrm{eff}}^{(m)}\right)
= O\!\left(\ell_\sigma D_{\sigma}^{\max}\,[\bar\alpha' + \log(1/\epsilon_{qsp})] + \mathrm{depth}(W_{\mathrm{sel}}(\hat H))\right),
\end{align}
\end{widetext}
up to architecture-dependent constant factors. An explicit mapping from $\epsilon_{\mathrm{tot}}$ to circuit depth requires specifying a gate-synthesis model that relates $\epsilon_{\mathrm{block}}$ and $\epsilon_{\mathrm{mux}}$ to two-qubit gate counts; for this reason, the main text reports parametric rather than hardware-specific expressions.

\subsection{Error-budget allocation}

The separation
\(
\epsilon_{\mathrm{tot}}
=
\epsilon_{\mathrm{factor}}
+
\epsilon_{\mathrm{block}}
+
\epsilon_{\mathrm{mux}}
+
\epsilon_{\mathrm{qsp}}
\)
introduced in Sec.~\ref{sec:composer} permits flexible allocation of error budget across algorithmic layers. While an equal split among contributions provides a convenient baseline, practical implementations may benefit from uneven distributions. For example, on hardware with native arbitrary-angle single-qubit rotations, it may be advantageous to assign a tighter tolerance to $\epsilon_{\mathrm{block}}$ while allowing a larger $\epsilon_{\mathrm{qsp}}$, or vice versa depending on routing constraints and noise characteristics.

The \texttt{COMPOSER} framework does not assume an optimal allocation a priori; rather, it exposes independent control knobs that can be tuned to minimize wall-clock runtime under hardware-specific constraints.

\subsection{Meaning of ``compile--once''}

As seen in Algorithm~\ref{alg:composer_compile_dial}, ``compile-once'' in this work denotes invariance of the logical two-qubit circuit topology under instance updates; it does not preclude hardware-level recompilation, pulse recalibration, or changes in fault-tolerant synthesis overhead. Subsequent updates to molecular geometry, active space, or truncation mask modify only single-qubit rotation angles and classical control data, without altering the two-qubit connectivity graph. This distinction is independent of the underlying fault-tolerant or NISQ execution model and should be understood as a circuit-scheduling concept rather than a claim about physical recompilation.

\begin{algorithm}[H]
\begin{algorithmic}[1]
\Statex \textbf{Compile stage} (once per orbital pool / qubit mapping):
\State Build rank-one ladder pools for $\hat H$ and (optionally) $\hat\sigma$ via nested factorizations; assign selector addresses.
\State Fix ladder pivots and routing schedule; synthesize the adaptor bank, \texttt{SELECT} wiring, and the reusable two-qubit fabric.
\State Choose global normalizations $(\alpha,\bar\alpha')$ and precompute the QSP phase list $\boldsymbol{\phi}$ for exponentiating the generator.
\State Hardware-map/optimize this two-qubit topology once; store the circuit skeleton with parameter slots.
\Statex
\Statex \textbf{Dial stage} (per geometry / mask / truncation update):
\State Recompute classical coefficients $\{\Omega_s\},\{\omega_s\}$ and choose masks $\mathcal M^{(m)}$ (and a model-space projector $P^{(m)}$).
\State Update only single-qubit rotation angles/phases in \texttt{PREP} and local ladder blocks (and, in FT settings, resynthesize rotations to the desired precision).
\State Execute the fixed skeleton to implement $W_{\mathrm{eff}}^{(m)}=U_{\hat\sigma}^{(m)\dagger} W U_{\hat\sigma}^{(m)}$ and downstream routines.
\end{algorithmic}
\caption{\texttt{COMPOSER}: compile stage vs dial stage (``compile-once, dial-later'').}\label{alg:composer_compile_dial}
\end{algorithm}

%%%%%%%%%%%%%%%%%%%%%%%%%%%%%%%%

\section{Diagnostics for Fixed Rank--One Algebraic Formulations}
\label{app:fixed_rank1}

This appendix provides supporting analysis for the mask construction and algebraic freezing assumptions used in \texttt{COMPOSER}. The material here is not required for the construction of the block-encoding pipeline, but serves to justify the stability of rank-one operator subspaces and the use of MP2-guided truncation strategies in adaptive similarity transformations.

\medskip
Within the unitary coupled-cluster doubles (UCCSD) approximation, the anti-Hermitian doubles generator may be written as in Eq.~\eqref{eq:op_sigma},
\begin{align}
\hat{\sigma}_2 = \sum_{s} \omega_s \bigl(L_s - L_s^\dagger\bigr),
\end{align}
{\color{black}
where each pair-excitation rank-one operator takes the (antisymmetric) form
\begin{align}
\hat{L}_s
&:= \hat B^\dagger[U^{(s)}]\;\hat B[V^{(s)}], \label{eq:Ls_pair_def}
\end{align}
with $B^\dagger[U^{(s)}]$ and $B[V^{(s)}]$ defined in Eq.~\eqref{eq:pair_ops}, and $U^{(s)}\in \wedge^2\mathbb C^{N_V}$ and $V^{(s)}\in \wedge^2\mathbb C^{N_O}$ the antisymmetric pair tensors ($U^{(s)}_{ab}=-U^{(s)}_{ba}$, $V^{(s)}_{ij}=-V^{(s)}_{ji}$). Here $\wedge^2\mathbb C^{N_V}$ denotes the antisymmetric pair space of dimension $\binom{N_V}{2}$. When desired, each antisymmetric tensor may be parameterized by two single-particle vectors via a wedge product, e.g., $U^{(s)}_{ab}=(x_a^{(s)}y_b^{(s)}-x_b^{(s)}y_a^{(s)})/\sqrt{2}$, and similarly for $V^{(s)}$, matching Sec.~\ref{subsec:cluster}.}

In practical applications, it is desirable to fix the algebraic structure defined by the vectors $\{U^{(s)},V^{(s)}\}$ and optimize only the scalar coefficients $\{\omega_s\}$. The following sections provide diagnostics and heuristics for assessing when this approximation is justified.

\subsection{Subspace stability diagnostics}

A natural diagnostic for algebraic stability is the robustness of the occupied and virtual subspaces under the similarity transformation. We therefore monitor the one-particle reduced density matrices
\begin{align}
(D^{\mathrm{occ}})_{ij}
&= \bra{\phi_0} e^{-\hat{\sigma}_2} a_j^\dagger a_i e^{\hat{\sigma}_2} \ket{\phi_0},
\label{eq:A_occ} \\
(D^{\mathrm{vir}})_{ab}
&= \bra{\phi_0} e^{-\hat{\sigma}_2} a_b^\dagger a_a e^{\hat{\sigma}_2} \ket{\phi_0},
\label{eq:A_vir}
\end{align}
where $\ket{\phi_0}$ is a reference determinant.

Stability of the algebraic formulation is indicated when the relative change in these density matrices remains small between successive updates,
\begin{align}
\frac{\|D^{\mathrm{occ,new}}-D^{\mathrm{occ,ref}}\|_F}
     {\|D^{\mathrm{occ,ref}}\|_F} \ll 1, \\
\frac{\|D^{\mathrm{vir,new}}-D^{\mathrm{vir,ref}}\|_F}
     {\|D^{\mathrm{vir,ref}}\|_F} \ll 1,
\end{align}
where $\|\cdot\|_F$ denotes the Frobenius norm. When these conditions are satisfied, the eigenbases defining $\{U^{(s)},V^{(s)}\}$ remain sufficiently invariant, justifying a fixed algebraic representation. If substantial subspace rotation is detected, the algebraic vectors may be recomputed periodically without altering the overall circuit topology.

\subsection{Subspace-overlap metric for rank-one manifolds (wAUC)}
\label{app:wAUC}

To compare MP2- and CCSD-derived \emph{rank-one} excitation manifolds, we reshape the doubles tensor into a matrix $T\in\mathbb C^{N_{V\!p}\times N_{O\!p}}$ with $N_{V\!p}=\binom{N_V}{2}$ and $N_{O\!p}=\binom{N_O}{2}$, indexed by antisymmetric pairs $A\equiv(a<b)$ and $I\equiv(i<j)$. Let the singular value decompositions be $T^{\mathrm{CCSD}}=\sum_{k} s_k\, u_k v_k^\dagger$ and $T^{\mathrm{MP2}}=\sum_{k} \tilde s_k\, \tilde u_k \tilde v_k^\dagger$. The corresponding rank-one operator basis vectors in the vectorized matrix space are $b_k := \mathrm{vec}(u_k v_k^\dagger)= v_k^*\otimes u_k$ and similarly $\tilde b_k$.

For a chosen rank $r$, define the orthonormal bases $B_r=[b_1,\ldots,b_r]$ and $\tilde B_r=[\tilde b_1,\ldots,\tilde b_r]$ and the subspace overlap
\begin{align}
\mathrm{ov}(r)
:=\frac{1}{r}\bigl\|B_r^\dagger \tilde B_r\bigr\|_F^2
=
\frac{1}{r}\sum_{j=1}^{r}\cos^2\!\theta_j,
\end{align}
where $\{\theta_j\}$ are principal angles between the two $r$-dimensional subspaces.

We then define a weighted average overlap (wAUC) up to rank $R$ as
\begin{align}
\mathrm{wAUC}(R)
:=\sum_{r=1}^{R} w_r\,\mathrm{ov}(r)
\end{align}
with $w_r:=\frac{s_r^2}{\sum_{k=1}^{R} s_k^2}$, so that leading singular components contribute most strongly (analogous to explained-variance weighting). In Sec.~\ref{sec:numerics}, $R=R_{\varepsilon}$ is determined by the relative screening criterion $s_r/s_1\ge \varepsilon_s$.

\subsection{MP2-guided construction of truncated rank--one operator lists}

In weak to moderate correlation regimes, the leading structure of the UCCSD doubles operator is well captured by second-order M{\o}ller--Plesset (MP2) theory. 
{\color{black}
These diagnostics are particularly relevant for Schrieffer--Wolff-style downfolding workflows, where one repeatedly updates generator coefficients (and masks/model spaces) while aiming to preserve a fixed operator manifold and execution topology.}

The MP2 doubles amplitudes
\begin{align}
t_{ab,ij}^{\mathrm{MP2}}
=\tfrac{1}{4}\cdot
\frac{\langle ij || ab \rangle}
     {\varepsilon_i+\varepsilon_j-\varepsilon_a-\varepsilon_b}
\label{eq:A_mp2}
\end{align}
provide a low-cost proxy for identifying dominant excitation patterns {\color{black}
(If one restricts sums to $i<j$ and $a<b$, the factor of $1/4$ is omitted; equivalently it may be absorbed into the definition of $t_{ab,ij}^{\mathrm{MP2}}$).}
Applying the same nested SVD--eigendecomposition used in the main text to $t_{ab,ij}^{\mathrm{MP2}}$ yields an initial rank-one operator list whose algebraic structure is typically stable under subsequent correlation refinement.

\paragraph{Ladder weights.}
{\color{black}
As a lightweight ranking heuristic (not a rigorous bound), we assign each rank-one ladder a scalar weight that correlates with its expected contribution to orbital-subspace mixing. A basis-invariant choice consistent with the pair-tensor form in Eq.~\eqref{eq:Ls_pair_def} is}
\begin{align}
w_s^{\mathrm{MP2}}
:=
|\omega_s|^{2}\,\|U^{(s)}\|_F^2\,\|V^{(s)}\|_F^2,
\end{align}
{\color{black}
where $\|\cdot\|_F$ denotes the Frobenius norm of the antisymmetric pair tensors (If $U^{(s)}$ and $V^{(s)}$ are normalized by construction, this reduces to $w_s^{\mathrm{MP2}}\propto|\omega_s|^2$). This weight is used only to rank candidate terms and to define cumulative coverage targets in the one-shot mask selection below.}

\paragraph{One-shot truncation algorithm.}
\begin{enumerate}
\item \textbf{Factorization.} Apply a pivoted approximate SVD to $t_{ab,ij}^{\mathrm{MP2}}$, retaining singular values above a threshold $\tau_{\mathrm{SVD}}$, and eigendecompose the resulting vectors with cut-off $\tau_{\mathrm{ED}}$ to obtain provisional rank-one operators.
\item \textbf{Weight evaluation.} Compute $w_s^{\mathrm{MP2}}$ for each term and sort the list by decreasing weight.
\item \textbf{Mask selection.} {\color{black}Choose the smallest subset $S$ such that the cumulative weight coverage satisfies $\sum_{s\in S} w_s^{\mathrm{MP2}}/\sum_{s} w_s^{\mathrm{MP2}} \ge \eta$ (e.g., $\eta=0.99$); the corresponding labels define a selector mask $\mathcal M$.}
\item \textbf{Optional validation.} After quantum optimization, recompute density matrices with tighter thresholds to verify that subspace deviations remain within tolerance.
\end{enumerate}

{\color{black}
In the benchmark set of Sec.~\ref{sec:numerics}, retaining only the largest $5$--$10\%$ of MP2-derived weights typically reproduces both occupied and virtual subspaces to within $10^{-3}$ in relative Frobenius norm, while maintaining a fixed selector width and circuit topology throughout adaptive iterations.}

%%%%%%%%%%%%%%%%%%%%%%%%%%%%%%%%%%%%%%

\bibliographystyle{unsrt}
\bibliography{ref}

\end{document}